\documentclass[journal,twoside,web]{ieeecolor}

\usepackage{generic}
\usepackage{cite}
\usepackage{amsmath,amssymb,amsfonts}

\usepackage{tabularray}
\usepackage{graphicx}

\usepackage{hyperref}

\usepackage{bbm}
\usepackage{leftidx}
\usepackage{extarrows}
\usepackage{colortbl}
\usepackage{CJK}
\usepackage{makecell}

\usepackage{stmaryrd}
\usepackage{enumerate}
\usepackage{stfloats}
\usepackage{amsfonts}
\usepackage{amssymb}
\usepackage{amsmath}
\usepackage{mathrsfs}

\usepackage{multirow}
\usepackage{diagbox}
\usepackage{booktabs}
\usepackage{bm}

\usepackage[T1]{fontenc}
\usepackage[utf8]{inputenc}
\usepackage{babel}
\usepackage[font=small,labelfont=bf,tableposition=top]{caption}
\usepackage{booktabs}
\usepackage{threeparttable}

\usepackage{cite}
\usepackage{subfigure}
\usepackage{graphicx}
\usepackage{tabularray}

\usepackage[T1]{fontenc}
\usepackage[utf8]{inputenc}
\usepackage{babel}
\usepackage[font=small,labelfont=bf,tableposition=top]{caption}
\usepackage{booktabs}
\usepackage{threeparttable}  

\usepackage{empheq}
\usepackage{makecell}

\usepackage{bbm}
\usepackage{leftidx}
\usepackage{extarrows}
\usepackage{colortbl}
\usepackage{CJK}
\usepackage{makecell}

\usepackage{multirow}
\usepackage{diagbox}
\usepackage{booktabs}

\usepackage{algorithm}
\usepackage{algorithmicx}
\usepackage{algpseudocode}

\newtheorem{Theorem}{Theorem}
\newtheorem{Lemma}{Lemma}
\newtheorem{ass}{Assumption}
\newtheorem{rmk}{Remark}
\newtheorem{defn}{Definition}
\newtheorem{proposition}{Proposition}

\hypersetup{hidelinks=true}
\usepackage{bm}
\usepackage{textcomp}
\def\BibTeX{{\rm B\kern-.05em{\sc i\kern-.025em b}\kern-.08em
    T\kern-.1667em\lower.7ex\hbox{E}\kern-.125emX}}
\begin{document}
\title{A Unified  Attack Detection Strategy for   Multi-Agent Systems over Transient  and  Steady Stages}
\author{Jinming~Gao,
	Yijing~Wang,
	Wentao~Zhang,~\IEEEmembership{Member,~IEEE},
	Rui~Zhao,~\IEEEmembership{Graduate Student Member,~IEEE}, Yang~Shi,~\IEEEmembership{Fellow,~IEEE}, and  
	Zhiqiang~Zuo,~\IEEEmembership{Senior Member,~IEEE}
\thanks{This work was supported by the National Natural Science Foundation of China
	No. 62173243,  No. 61933014.}
\thanks{				J.  Gao, Y.  Wang, R. Zhao and Z. Zuo are with the Tianjin Key Laboratory of Intelligent Unmanned Swarm Technology and System, School of Electrical and Information Engineering, Tianjin
	University, 300072, P. R. China.
	(e-mail: gjinming@tju.edu.cn; yjwang@tju.edu.cn; ruizhao@tju.edu.cn;  zqzuo@tju.edu.cn)}
\thanks{			W. Zhang is with the Continental-NTU Corporate Lab, Nanyang
	Technological University, 639798, Singapore, and is also with the School
	of Electrical and Electronic Engineering, Nanyang Technological University,
	639798, Singapore.  (e-mail: wentao.zhang@ntu.edu.sg)}
\thanks{Y. Shi is with the Department of Mechanical Engineering, University of Victoria, Victoria, BC V8W 2Y2, Canada. (e-mail: yshi@uvic.ca)}}

\maketitle

\begin{abstract}
This paper proposes  a unified  detection  strategy  against three kinds of attacks   for  multi-agent systems (MASs) which is applicable to both transient  and  steady stages.  For attacks   on the  communication  layer, a watermarking-based detection scheme with {Kullback-Leibler (KL)} divergence is designed. Different from traditional communication schemes, each agent transmits a message set containing  two state values with different types  of watermarking.     It is found that the detection performance is determined by the relevant parameters of the watermarking signal. Unlike the existing detection manoeuvres,  such a scheme is capable of  transient and steady stages. For attacks  on  the   agent layer, a convergence rate related detection
approach   is put forward. It is shown that the resilience of the  considered system is characterized  by the coefficient and offset of the envelope. For hybrid attacks, based on the above detection mechanisms, 
a  general framework  resorting to trusted  agents is presented, which requires  weaker  graph conditions  and less information transmission. Finally, an  example associated  with the platooning of  connected vehicles is  given to support the theoretical results.
\end{abstract}


\section{Introduction}
\label{sec:introduction}
{Due}  to the development of communication  and computer technologies, significant progress has been made in the field of  cyber-physical systems (CPSs). This, however, makes them vulnerable to cyber  attacks. The reasons lie in that, unlike traditional information security, the  security issues in CPSs  display several
remarkable features: 1) The underlying plants are  spatially distributed; 2) The types of attacks launched are diverse even complex; 3) New and unknown attacks are emerging. These   facts have been greatly confirmed 
in  security incidents, for example,  Stuxnet virus invaded Iran's nuclear facilities in 2010 \cite{mo2015physical}, Havex virus attacked SCADA  resulting in disabled hydropower dams in 2014 \cite{maitra2015offensive}.

Existing literature  on the  security of CPSs can be divided into two categories. The  first  studies   the vulnerability and performance limit with the aim of launching  a more cunning attack \cite{zhang2020false,lu2022false,zhang2015optimal}. By contrast, the second focuses  on mitigating the negative effect of attacks, including 
attack mitigation  \cite{kwon2016cyber}, secure state estimation \cite{fawzi2014secure},  and  resilient control \cite{yuan2016resilient}. Apart from  that, attack detection is also an active research field   \cite{pasqualetti2013attack,lu2022detection,carvalho2018detection}. {There have been  numerous   fault/attack detection mechanisms proposed here as well \cite{zhang2010fault,gallo2020distributed,barboni2020detection}.} 
{Additionally, watermarking is commonly utilized in detection work as a supplementary technology that has been extensively employed in industrial settings\cite{yang2022joint,ahmed2022practical}. 
The existing watermarking-based detection strategies  can be divided into two categories: the additive watermarking and the multiplicative watermarking. For the former, a  physical watermarking approach was first  presented  to expose replay attacks by $\chi ^2$ detector   \cite{mo2009secure}.  Furthermore, in order to better quantify the detection effect, watermarking Kullback-Libeler (KL) divergence detector was  proposed in   \cite{mo2009secure,mo2015physical}. In order to reduce   control cost caused by the involved   watermarking, an optimal watermarking  scheduling strategy  \cite{fang2020optimal} and a quickest detection through  parsimonious watermarking protocol  \cite{naha2023quickest} were suggested. To handle   cunning attacks and even   collaborating attacks, some  complex procedures were designed.    A multi-channel watermarking cooperative defense strategy joint with  moving target ideas   has been proposed in \cite{ghaderi2020blended}. In addition, \cite{liu2023proactive} developed  a unified watermarking-based detection framework against  a wider range of attack types. 
	Besides, multiplicative watermarking has also  been  employed  to  tackle the drawbacks of the additive one \cite{ferrari2020switching,ferrari2017detection,teixeira2018detection,ferrari2021detection}.  
 It is noted that the additive watermarking technology itself  has some  shortcomings in defending against  stealthy additive  attacks, and    multiplicative watermarking has  been  presented  to tackle the drawbacks of the additive one.
	In  \cite{ferrari2017detection}  and   \cite{teixeira2018detection}, the sensor outputs were  encrypted with the help of   impulse response filters, and then  were decrypted by equalizing filters to relieve the  performance degradation caused by watermarking. In \cite{ferrari2020switching}, a switching  watermarking filter framework  was   devised to further increase  the confidentiality of  relevant   parameters to the attacker. Recently, a secure estimation issue was also  addressed  in \cite{ferrari2021detection}. It should be noted that in the existing related work, the watermarking  technique is often utilized with estimator/observer, which belongs to the  model-based strategy.}

On the other hand, multi-agent systems (MASs) are featured by inherent properties in CPSs and have witnessed wide  applications in various engineering areas, such as intelligent traffic systems\cite{10014016}, unmanned aerial vehicles  \cite{yu2019distributed}, smart grid systems\cite{10018476} and multi-sensor network\cite{9983527}.  
Due to the dependence on network and computer technology, the properties such as openness, spatial distribution and inter-connection on MASs pave a way to  malicious attacks  within different layers  including control layer,   communication  layer and agent layer \cite{he2021secure}. Typical effort includes but is not limited to Worm virus   \cite{he2021secure}, denial of service (DoS) attacks  \cite{zuo2021resilient} and false data injection attacks  \cite{8840898}. When targeted on the  controller layer, the vulnerability of  MASs can be aroused by malicious computer malware  via implanting some computer worm \cite{he2021secure}. 	
For attacks on the  communication  layer, several   defense mechanisms have been proposed.  With elaborate design of the  attack signal,  it can  destroy  data integrity and  instrumental detection strategies \cite{shames2011distributed,sundaram2010distributed,pasqualetti2011consensus,2022198}. Besides, on account of  the  distributed characteristics, each agent in MASs has multiple information sources. Therefore, additional detection approaches, for example,    multi-hop communication \cite{yuan2021secure} and consensus-based attack detection strategy \cite{mustafa2020resilient} were introduced. 
{A common feature of the aforementioned work is that they are all  built on
	the observer/consensus-based framework. That is to say, the successful defense against attacks  strongly depends
	on the convergence of estimation error which  does not occur during the   transient stage. }

If the attacks  act  on the  agent layer,  they   mainly affect the  consensus of interacting agents. One of the representatives is the  Byzantine attack  \cite{ishii2022overview,leblanc2013resilient,wu2020federated}.  Many  defense strategies were  resorted to   network robustness-based mean-subsequence-reduced (MSR) algorithm. This algorithm  is inspired by  the idea of  ignoring neighboring agents that may have  extreme state values  \cite{ishii2022overview}.  \cite{leblanc2013resilient} designed a weighted-mean-subsequence-reduced (W-MSR) algorithm for  Byzantine attacks. Furthermore, \cite{yan2022resilient} further  expanded the above  results   to  multi-dimensional systems.   	{However, the  robustness of these findings is  somewhat conservative. Actually, several  normal neighbors may also be excluded,  resulting in unnecessary defense cost. In order to locate malicious agents, an  isolation-based approach was proposed  in \cite {9415167}. Nevertheless,  the  trade-off between detection accuracy and system resilience has not been fully investigated. }

{
	By contrast,   the research on hybrid  attacks   is more challenging \cite{zhang2021physical}.
	 Since  Byzantine attacks on the   agent  layer can be regarded   as a special case in the   attack on the  communication  layer from  the  receiving side  as pointed out  in  \cite{fu2019resilient}, the performance of  two kinds  of  attacks is similar for  the receiving agent, causing  additional  difficulty  in  distinguishing  them. Thus, the individual or simple combination of detection strategies will fail.  
	It is noted that  the effective identification of the attack categories helps  to take  more precise countermeasures in the subsequent procedures for communication  or agent layers, thereby  the     reducing defense cost.	 Unfortunately, there are currently few  relevant   work. }

{We  emphasize that for    attacks  on the  communication layer,  no matter whether it is  the  traditional attack detector, the watermarking  detection mechanism, or the consensus-based detector,  they are all  built on the observer/consensus-based framework \cite{zhou2022watermarking,2022198,mustafa2020resilient}. That is to say,  the  successful execution of the detector strongly depends on the  condition that the   estimation error or consensus error is zero. Therefore, they are  
inapplicable  to the transient stage. This in turn  gives an opportunity for  the attackers, especially  the systems involving  frequent dynamic adjustment such as vehicle formation  \cite{bian2019reducing}. In particular,  when the attacks are injected during the transient stage, the observer-based detector will become  invalid. 
For   attacks  on the  agent layer, the MSR  and its variants \cite{leblanc2013resilient,lu2023bipartite,ishii2022overview} may cause error isolation on normal agents, leading to  increased  defense cost. Then it is necessary to implement  attack  detection      in MASs. The existing detector-based isolation work focuses on  recognizing   attacked ones accurately to pave a way for resilient consensus \cite{9415167,2022198}. However,  from the perspective of achieving resilient  consensus,   there is no need to adopt a   zero tolerance attitude towards   malicious agents  that do not affect the convergence of  MASs. 
Moreover, when facing  hybrid attacks  composed of the above two kinds of attacks, the issue under consideration  becomes  more challenging for this situation \cite{zhang2021physical}. The Byzantine attacks on the agent layer can be seen as a specific instance of attacks on the communication layer, particularly from the receiving side  \cite{fu2019resilient}, causing  further hinder  in  distinguishing  them. Besides,   more precise defensive action in the subsequent procedures  for communication  or agent layers can reduce defense cost.
Therefore, it is necessary to establish a unified  detection framework for MASs  over different layers which are valid in  both transient and steady  stages. This is another important motivation for conducting this research.}
This paper studies   the attack detection scheme for MASs. Specially, for attacks  on the  communication  layer, a KL divergence detection scheme is proposed to  transmit the   message set equipped with two  kinds of watermarking. This  extends the detection range to the  transient stage. For attacks  on  the  agent layer,   a detection  scheme based on the convergence rate to prevent normal agents from being misjudged is put forward. Compared with MSR scheme, the proposed method enables to measure a trade-off between resilience and detection performance. Finally,  a unified detection scheme  is presented  for hybrid  attacks. By doing so, we can accurately locate and distinguish  hybrid attacks   when they   occur.

{The main contributions of this paper are summarized as follows:	}

\begin{enumerate}
	\item  {  A watermarking-based detection strategy is proposed with   KL divergence. To implement detection during the  transient stage, each agent transmits a message set  containing two state values. And each  state value is equipped  with multiplicative and additive watermarking  of different parameters. In this way,  the attacks   on the  communication layer   can be detected with the help of   KL divergence after deleting watermarking. A sufficient condition is derived to ensure the detection performance. In contrast to  the existing results of \cite{2022198} and \cite{mustafa2020resilient}, no estimation process or consensus process  is employed. Therefore, the proposed detection scheme is valid for    the  transient stage.}
	
	\item { For attacks  on the  agent layer, an  envelope-based detector is proposed in our work.  Unlike  \cite{leblanc2013resilient} and \cite{  ishii2022overview}, the false isolation of normal agents can be efficiently avoided by employing   a   detector. Moreover, the detector  can tolerate the misbehaving agents whose   state values  are  contained in a  monotonically decreasing  envelope about tracking error. In other words, it can handle  the attacks   that do not affect system convergence. 
		Hence, the proposed  detector can   facilitate  system  resilience while ensuring  stability.}
	
	\item { A  detection scheme for  hybrid  attacks   guaranteeing  transient and steady stages is suggested. This  scheme is  constructed in terms of  the  trusted agents and two-hop communication.  Specifically,  the  robustness of the graph regarding  two-hop communication is utilized  to guarantee that there exists  at least one  trusted agent in  each pair of agents to 
		distinguish situations of hybrid attacks.  Compared with \cite{yuan2021secure}, the proposed method has a moderate graph requirement and less information transmission. 
}   
\end{enumerate}

The outline of this paper is as follows. In Section ${\rm \ref{GOsec2}}$, some preliminaries are  given, including    notations, graph theory and  system description. The detection strategies for three kinds of attacks    and their performance analysis   are  presented in Section ${\rm \ref{GOsec3}}$.  Simulation results  about the platoon of connected vehicles are  provided  in Section ${\rm \ref{GOsec4}}$. Finally, Section ${\rm \ref{GOsec5}}$ concludes this paper.

\section{Preliminaries}\label{GOsec2}

\subsection{Notations and Graph Theory}

Denote $\mathbb{R}^{n}$   the set of  $n$-dimensional real vectors and $\mathbb{R}^{n\times m}$   the set of  $n\times m$-dimensional real matrices. $\mathbb{R}^{+}$ is  the set of positive constants. ${\rm diag}\left\{a_{1},\cdots,a_{n}\right\}$ represents a diagonal matrix. $\bm{I}$ and  $\bm 1$  are    the   identity  matrix and the vector whose all  elements being 1, respectively. $\left\| \cdot \right\|$ is   the Euclidean norm.

Let $\mathcal{G}=(\mathbb{E},\mathbb{V})$ be a directed graph containing  $N+1$ nodes, where $\mathbb{E}$ is the set of edges and $\mathbb{V}$ is the set of nodes.  A path from  $v_{m}$ to $v_{n}$ is a
sequence of distinct nodes $(v_{m}, v_{m,1},  v_{m,2}, . . . ,v_{n})$, where $ (v_{m,j}, v_{m,{j+1}})\in \mathbb{E} $ for $j = 1, ..., l-1$. Such a  path is also referred to an $l$-hop path.  $N_i^{+}$ and $N_{i}^{-}$ are    sets of  in-neighbors and out-neighbors for  agent $i$.  A  weighted adjacency matrix is defined as  $\mathscr{A}=[a_{ij}]\in \mathbb{R}^{N\times N}$ and  $a_{ij}>0$ if $(v_j, v_i) \in \mathbb{E}$,  $a_{ij}=0$  otherwise. The Laplacian matrix of $\mathcal{G}$ is $L=[l_{ij}]\in \mathbb{R}^{N\times N}$, where $l_{ii}=\sum\limits_{i=1,j \neq i}^{N} a_{ij}$ and $l_{ij}=-a_{ij}$, for $i \neq j$.

	\subsection{System  Description}
{Consider a group of $N+1$ leader-following agents   labeled as  $\left\{ 0,\cdots ,N \right\}$, in which agent  $0$  is the leader and the others are the followers. Similar to  \cite{wang2015consensus}, when there is no attack, the discrete-time dynamics of the $i$-th normal agent in  $\left\{ 0,\cdots ,N \right\} $ is }

\begin{equation}\label{GOeq1}
	\begin{aligned}
		{x}_i\left( k+1 \right) =A {x}_i\left( k \right) +B{u}_i\left( k \right),~i=0,\cdots ,N,
	\end{aligned}
\end{equation} 
where  {the time $k$  is the discrete-time index,} ${x}_i\left( k \right) \in \mathbb{R} ^n$ is the state vector, ${u}_i\left( k \right) \in \mathbb{R} $ is the control input and
\begin{equation*}
	\begin{aligned}
		A=\left[ \begin{matrix}
			0&		1&		\cdots&		0\\
			\vdots&		\vdots&		\ddots&		\vdots\\
			0&		0&		\cdots&		1\\
			\rho  _1&		\rho  _2&		\cdots&		\rho  _n\\
		\end{matrix} \right] \in \mathbb{R} ^{n\times n}\,\,\mathrm{and}~ B=\left[ \begin{array}{c}
			0\\
			\vdots\\
			0\\
			1\\
		\end{array} \right] \in \mathbb{R} ^n,
	\end{aligned}
\end{equation*} 
with  $\left\{ \rho _1,\cdots ,\rho _n \right\}$ being the  coefficients determined by the agent's 
dynamic characteristics.

To achieve  leader-following consensus in the  mean square sense,  the   control protocol  can be designed as 

\begin{equation*}\label{GOq2}
	\begin{aligned}
		{u}_i\left( k \right) =K_1{x}_i\left( k \right) +a_i\left( k \right) \sum_{j\in \mathcal{N} _i}{a  _{ij}K_2\left( \widetilde{y}_{ij}\left( k \right) -{x}_i\left( k \right) \right)},
	\end{aligned}
\end{equation*}
where  $ \widetilde{y}_{ij}\left( k \right)={x}_{j}\left( k \right)+{w}_{ij}\left( k \right)$ denotes the information that  agent $i$ receives from
agent $j$. Here ${w}_{ij}\sim N\left( 0,\varSigma _1 \right)$ is  the Gaussian  noise in edge $(j,i)$ with $\sigma _{1}^{2} \in \mathbb{R}^{+}$  the diagonal element of  $\varSigma_1 \in \mathbb{R}^{n\times n}$.  The controller gains   $K_1\,\,=\,\,\left[ -\rho_1+b_1,-\rho_2+b_2-b_1,\cdots ,-\rho_{n-1}+ b_{n-1}-b_{n-2},-\rho_n  \right.\\ \left. -b_{n-1}+1 \right] \in \mathbb{R} ^{1\times n}$, $K_2\,\,=\,\,\left[ b_1,b_2,\cdots ,b_{n-1},1 \right] \in \mathbb{R} ^{1\times n}$ and $a_i\left( k \right) \in \mathbb{R}  >0$ is the time-varying noise-attenuation gain,  please see  \cite{wang2015consensus} for  details.

Now we give some assumptions, which will be used throughout this paper. 

\begin{ass} \hspace{-0.001cm}\cite{wang2015consensus}
	\label{GOas1} 
	
	(A1-1)  The communication
	graph  has a spanning tree;
	
	(A1-2) $\sum\limits_{i=0}^\infty a(k)=\infty$ and $\sum\limits_{i=0}^\infty a^2(k)<\infty$;

	(A1-3) All roots of $s^{n-1}+b_{n-1}s^{n-2}+\cdots +b_2s+b_1=0$ are inside the unit circle.
\end{ass}

\begin{ass}\label{GOas2} \hspace{-0.001cm}\cite{ishii2022overview,AN2021109664}
	($(L,P)$-local attack model) For each agent, there  are at most $L$ misbehaving agents in the in-neighbors of any agent and at most $P$  malicious in-communication channels.
\end{ass} 

\begin{ass}\label{GOas3}
	The range of  states in  normal systems is   disclosed to the attacker, that is, $\epsilon _1=\underset{l\in \left\{ 1,\cdots ,n \right\}}{\min}\left\{ \underset{i\in \mathbb{V}   }{\min}x_{i,~l}\left( k \right) \right\} 
	$ and $\epsilon _2=\underset{l\in \left\{ 1,\cdots ,n \right\}}{\max}\left\{ \underset{i\in \mathbb{V}}{\max}x_{i,~l}\left( k \right) \right\} 
	$ are available  to  the attacker.
\end{ass} 

\begin{rmk}
	In \emph{Assumption} $\ref{GOas2}$,  $L$ and $P$ are known and  determined by  the  prior knowledge of the system, including the topological structure, the vulnerability of  agents  and communication channel \cite{gracy2021security}. Due to the impact of control objectives  and hardware constraints, the state information  of  normal system is  bounded, and will  be  accessible  by the   attacker \cite{carvalho2018detection}. Furthermore, if  system output caused by the attack exceeds the  normal range in reality, the attack can be easily exposed. Therefore, in order to keep stealthy, the attacker will keep  the output signal within a  normal range \cite{pasqualetti2013attack}. More detailed  results will be discussed in the sequel.	
\end{rmk}

Under \emph{Assumption} $\ref{GOas1}$, if the underlying graph of system $(\ref{GOeq1})$ without any  attack contains a spanning tree, the MASs can achieve consensus, i.e.,
\begin{equation*}\label{GOeq2}
	\begin{aligned}
		&\underset{k\rightarrow \infty}{\lim\limits}\mathbb{E}\left\| {x}_i\left( k \right) -{x}_0\left( k \right) \right\| ^2=0,\\
		&\underset{k\rightarrow \infty}{\lim\limits}{\rm sup}~\mathbb{E}\left\| {x}_i\left( k \right) \right\| ^2<\infty. 
	\end{aligned}
\end{equation*}

Note that most  existing findings are only valid for  the steady stage, that is, the time required for the signal  $\eta(k)\triangleq\underset{i\in \mathbb{V}\backslash\{0\}}{\max}\left\{\mathbb{E}\left( \frac{\left\| \left( x_i\left( k \right) -x_0\left( k \right) \right) \right\|}{\left\| x_0\left( k \right) \right\|} \right) \right\}
$ \cite{franklin2002feedback}  to reach and stay within a given tolerance range thereafter,  with the help of estimator \cite{9415167,lu2022detection} or consensus value \cite{mustafa2020resilient}. Such a  scenario is not applicable in practice, especially for multiple dynamic adjustment processes, such as vehicle formation  with leaders \cite{bian2019reducing}. Thus, 
the goal of this paper is  to develop a detection framework for MASs against  attacks on  different layers while guaranteeing its effectiveness in  both  transient and steady stages. 

	\section{Main Results}\label{GOsec3}
Three kinds of attacks   shall be addressed in this part, including  attacks on the    communication  layer, attacks on  the  agent layer, and  hybrid attacks.   

\subsection{Detection  Strategy for Attacks on the   Communication   Layer}\label{opp1}

In this subsection,  we  focus on  the attacks that tamper  data $\widetilde{y}_{ij}(k)$ transmitted in the communication channel between two agents. {Without loss  of generality,  we set the attack strategy on edge $(j,i)$ as    $\widetilde{y}^{a}_{ij}(k)=\varXi _{ij}(k) \widetilde{y}_{ij}(k)+\varLambda_{ij}(k)$  where $\varXi  _{ij}\left( k \right) \in \mathbb{R}^{n\times n}$  is a diagonal matrix with $\varXi  _{ijl}\left( k \right)$ being  its $l$-th diagonal element and $\varLambda_{ij}(k) \in \mathbb{R}^{n}$ with $\varLambda  _{ijl}\left( k \right)$ being  its $l$-th element. It is worth emphasizing that $\varXi _{ij}(k)$ means multiplicative attack and $\varLambda_{ij}(k)$ corresponds to  additive attack.} 

For  MASs, attacks need to  keep stealthy to prevent them from being detected so that  system output operates like a  normal one \cite{2022198}. Thus, according to \emph{Assumption} $\ref{GOas3}$, the  attacker aims to design  a  signal guaranteeing the received state information by normal agent  staying  within the normal range, that is, $\mathbb{E}(\widetilde{y}^{a}_{ijl}(k)) \in \left[ \epsilon _1,\epsilon _2 \right] $, where $\widetilde{y}^{a}_{ijl}(k)$ is the $l$-th entry. Thus, we give  the following claim.  

\begin{proposition}\label{GOpro1}
	For the attacked system $(\ref{GOeq1})$,  $\mathbb{E}(\widetilde{y}^{a}_{ijl}(k)) \in \left[ \epsilon _1,\epsilon _2 \right]$ if and only if  
	\begin{equation}\label{GO0q222}
		\begin{aligned}      
			&\varXi _{ijl}(k) \in \left[ -1,1 \right], \\
			&\varLambda_{ijl}(k) \in \left\{ \begin{array}{l}
				\left[ 0, \epsilon _1+\epsilon _2 \right] , \epsilon _1>0\\
				\left[ \epsilon _1+\epsilon _2,0 \right] , \epsilon _2<0\\
				\left[ \epsilon _1,\epsilon _2 \right] , \epsilon _1<0<\epsilon _2.\\
			\end{array} \right. 
		\end{aligned}
	\end{equation}	 
\end{proposition}

\begin{proof}
	Note that   $\mathbb{E}(\widetilde{y}^{a}_{ijl}(k)) \in \left[ \epsilon _1,\epsilon _2 \right]$ is equivalent to	
	\begin{equation*}\label{GOewq3}
		\begin{aligned}
			\left\{ \begin{array}{c}
				\epsilon _1\leqslant \varXi _{ij}(k) \epsilon _1+\varLambda_{ij}(k) \leqslant \epsilon _2\\
				\epsilon _1\leqslant \varXi _{ij}(k) \epsilon _2+\varLambda_{ij}(k) \leqslant \epsilon _2.\\
			\end{array} \right. 
		\end{aligned}
	\end{equation*}
	
	Through  linear programming, the proof is evident.   \end{proof}

To handle the noise,   KL divergence detector and $
\chi ^2-$ detector are common  choices. It is known  that  the former can tackle    disturbance, where    $\chi^2$-detector fails.  Besides, KL divergence   provides  a quantitative measure  on  the distance between two distributions  \cite{mo2015physical}. Thus, the following concept about KL divergence is given. 

{\begin{defn}\label{GOde1}
		Let $f_{a}(\gamma)$ and $f_{b}(\gamma)$ be the probability density functions (PDFs) of $a$ and $b$, the  KL divergence $D_{KL}(a \Vert b)$ between $a$ and $b$ is  
		\begin{equation}\label{GOeq3}
			\begin{aligned}
				D_{KL}(a\parallel b)\triangleq \int_{\left\{ \left. \gamma \right|f_a(\gamma )>0 \right\}}{f_a(\gamma )\log \frac{f_a(\gamma )}{f_b(\gamma )}d\gamma},  
			\end{aligned}
		\end{equation}
		where $\gamma \in \mathbb{R}$  is an  integral variable. For  simplicity, the item   ${\left\{ \left. \gamma \right|f_a(\gamma )>0 \right\}}$ is omitted in the following discussions. 
\end{defn}}

A KL divergence-based detection scheme was proposed to study the resilient  consensus in  \cite{mustafa2020resilient}. However, the scheme becomes  invalid if an  attack is injected in the transient stage, see Fig. $\ref{GOfig23}$ for  details. {Specifically, we set the time  transient stage as the time interval that   the state value of each follower satisfies $\eta(k)\triangleq\underset{i\in \mathbb{V}\backslash\{0\}}{\max}\left\{\mathbb{E}\left( \frac{\left\| \left( x_i\left( k \right) -x_0\left( k \right) \right) \right\|}{\left\| x_0\left( k \right) \right\|} \right) \right\} \geqslant \varsigma$, where $\varsigma \in \mathbb{R}$ is a small  scalar which is  usually chosen as $0.01, 0.02$ or  $0.05$ \cite{franklin2002feedback}. It  corresponds to  {$k\leqslant  14$}, which is  indicated by  the left side of  black dash-dot  line in Fig. $\ref{GOfig23}$.} From Fig. $\ref{GOfig23}$,  it is seen  that  during the transient stage, the KL divergence cannot be guaranteed to stay  within the safe range. As a result, the  detector's alarm  will be  triggered incorrectly.

\vspace{-1.3cm}

\begin{figure}[H]
	\centering
	\includegraphics[width=3.2in,height=2.4in]{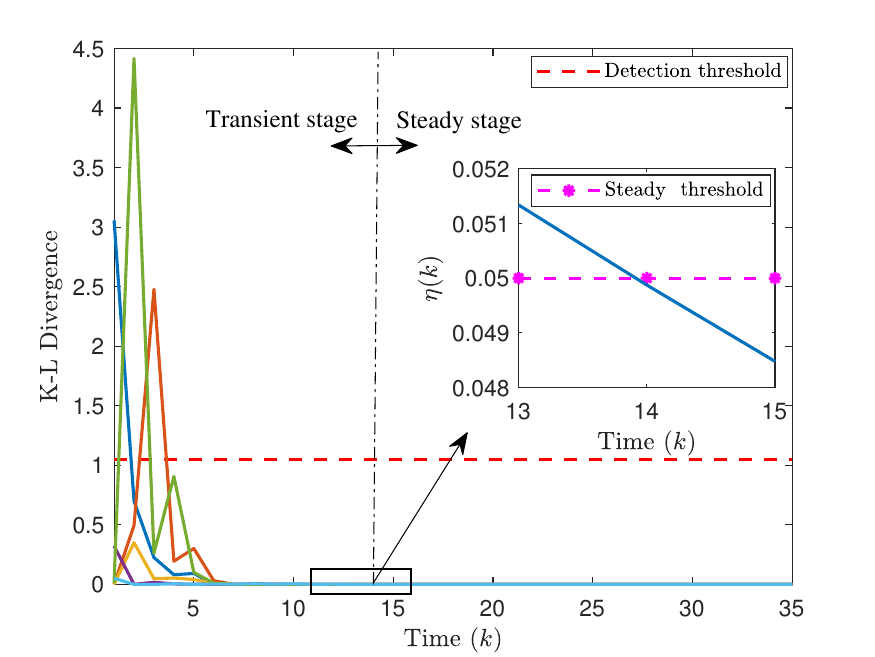}
	\caption{KL divergences of detection strategy in \cite{mustafa2020resilient}.}
	\label{GOfig23}
\end{figure}

\vspace{-1.3cm}

To avoid the defect  incurred by KL divergence, a watermarking-based detection scheme is put forward  for both transient
and steady stages in this  paper, see  Fig. \ref{GOfig14}. More specific, each agent   transmits a message set consisting of  two state values  with multiplicative watermarking $M_1^{-1}(k)\in \mathbb{R} ^{n\times n}$,  $M_2^{-1}(k) \in \mathbb{R} ^{n\times n}$ and additive watermarking $F_1(k) \in  \mathbb{R} ^{n}$, $F_2(k) \in  \mathbb{R} ^{n}$. The key  idea   is to measure  the difference between  two values  to determine whether the data has been tampered. 
To this end, let   {$M_{1}^{-1}(k)\triangleq \lambda_1\bm{I}+\mathrm{diag}\left\{M_{11}^{2}(k), \cdots ,M_{1l}^{2}(k),  \cdots, 	M_{1n}^{2}\left( k \right) \right\}$ with  $ l \in \left\{ 1,\cdots ,n \right\}$,}  where $M_{1l}\left( k \right) \sim \mathcal{N} \left( 0,\sigma _{M_{1}}^{2} \right)$ and   $\lambda_1 \in \mathbb{R}$  is a positive constant. Similarly,  {$M_{2}^{-1}(k)\triangleq \lambda_2\bm{I}+\mathrm{diag}\left\{ M_{21}^{2}(k),\cdots  ,M_{2l}^{2}(k),\cdots ,M_{2l}^{2}\left( k \right) \right\}$} with  $M_{2l}\left( k \right) \sim \mathcal{N} \left( 0,\sigma _{M_{2}}^{2} \right)$. 
$F_{1}(k) \in \mathbb{R}^{n}$  is  a random vector where the $l$-th element satisfies $F_{1l}(k)\sim \mathcal{N} \left( 0,\sigma_{F_{1}}^{2} \right)$ and $F_{2}(k)$ is equipped with the same structure and  $F_{2l}(k)\sim \mathcal{N} \left( 0,\sigma_{F_{2}}^{2} \right)$. In this way, the attack model on edge $(j,i)$ becomes    $\overline{y}^{a}_{1ij}(k)=\varXi _{1ij}(k) \overline{y}_{1ij}(k)+\varLambda_{1ij}(k)$ and $\overline{y}^{a}_{2ij}(k)=\varXi _{2ij}(k) \overline{y}_{2ij}(k)+\varLambda_{2ij}(k)$   where $\varXi  _{rij}\left( k \right) \in \mathbb{R}^{n\times n}$, $\varLambda _{rij}\left( k \right) \in \mathbb{R}^{n}$  $\left(r\in \left\{ 1,2 \right\} \right)$ are the attack signals.   
The defense steps    are summerized in  \emph{Algorithm}  $\ref{alg1}$ and  its  framework is given  as below.

The attack detection suggested  by  hypothesis test \cite{zhou2022watermarking} has the form $D_{KL}(y_{1ij}^{*}(k) \Vert  y_{2ij}^{*}(k))\underset{H_1}{\overset{H_0}{\lessgtr}}\theta$, where $\theta$ is a  pre-set threshold.  $H_0$ and $H_1$  imply
that  edge $(j,i)$ is secure or under attack.

\begin{figure}[H]
	\centering
	\includegraphics[width=3.5in,height=1.5in]{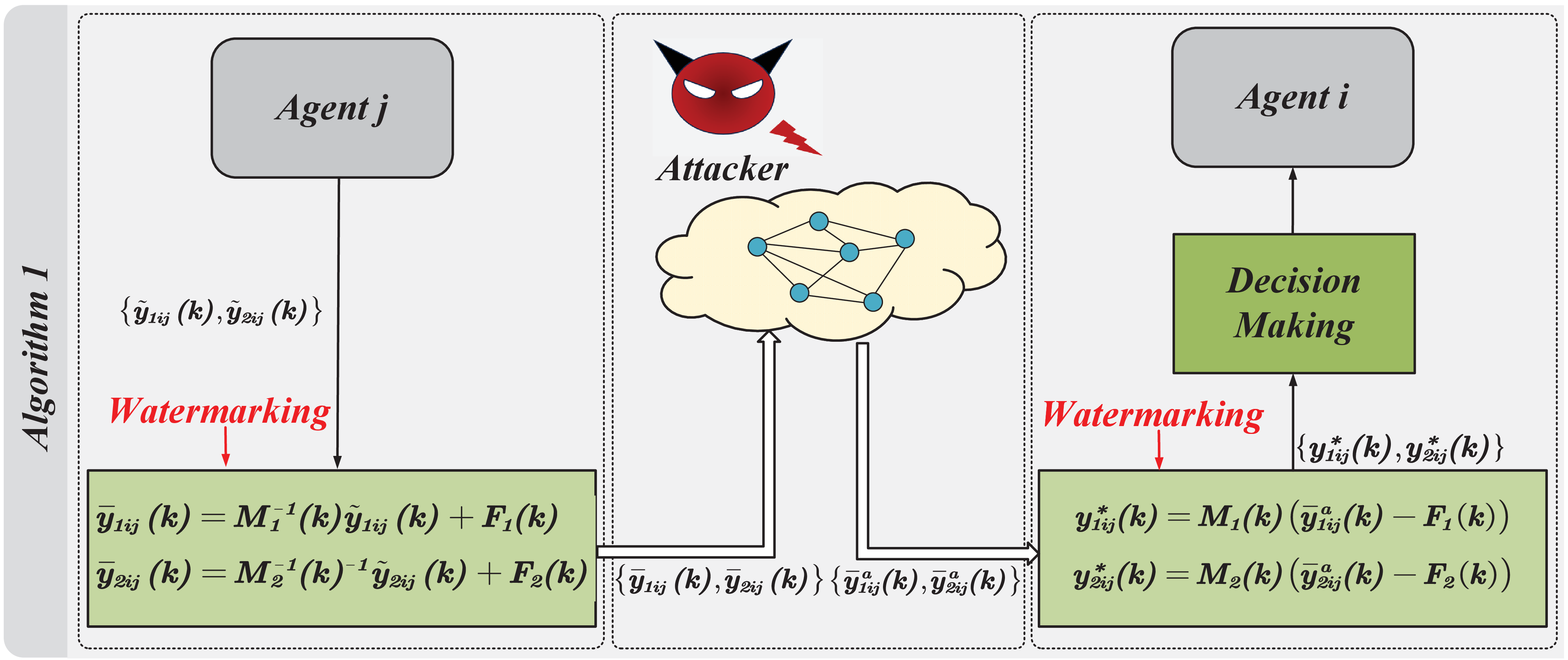}
	\caption{Block diagram of   \emph{Algorithm}  $
		\ref{alg1}$ for  attacks  on the  communication  layer.}
	\label{GOfig14}
\end{figure}

\begin{algorithm}[H]
	\caption{ Watermarking-based Detection Strategy for Attacks on the  Communication  Layer}
	\label{alg1}
	\begin{algorithmic}[1]
		\State{For  agent  $i$,  each agent $j \in {N_{i}^{+}} $  generates   the message set $\left\{\widetilde{y}_{1ij}(k),~\widetilde{y}_{2ij}(k)\right\}$};
		\State{The message set $\left\{\widetilde{y}_{1ij}(k),~\widetilde{y}_{2ij}(k)\right\}$} is equipped with watermarking as  	
		\begin{equation}\label{bbp1}
			\begin{aligned}
				&\overline{y}_{1ij}(k)=M_{1}^{-1}(k)\widetilde{y}_{1ij}(k)+F_1(k),\\
				&\overline{y}_{2ij}(k)=M_{2}^{-1}(k)\widetilde{y}_{2ij}(k)+F_2(k);
			\end{aligned}
		\end{equation}	
		\State{ The message set $\left\{ \overline{y}_{1ij}\left( k \right) ,\overline{y}_{2ij}\left( k \right) \right\}$ is transmitted  into communication edge $(j,i)$};
		\State{ Once  $\left\{ \overline{y}^{a}_{1ij}\left( k \right) ,\overline{y}^{a}_{2ij}\left( k \right) \right\}$ has been  received, the watermarking is removed by performing 
			\begin{equation}\label{GOeq22op2}
				\begin{aligned}
					&y_{1ij}^{*}(k)=M_1(k)\left( \overline{y}^{a}_{1ij}(k)-F_1\left( k \right) \right),\\
					&y_{2ij}^{*}(k)=M_2(k)\left( \overline{y}^{a}_{2ij}(k)-F_2\left( k \right) \right);
				\end{aligned}
			\end{equation}		
		}
		\State Decision making: Calculate $D_{KL}(y_{1ij}^{*}(k) \Vert  y_{2ij}^{*}(k))$.
		\If {$D_{KL}(y_{1ij}^{*}(k) \Vert  y^{*}_{2ij}(k))\leqslant \theta$}
		\State {Communication channel $(j,i)$ is  free of attack;}
		\Else ~{$D_{KL}(y_{1ij}^{*}(k) \Vert  y_{2ij}^{*}(k))> \theta$}		
		\State {Communication channel $(j,i)$ is attacked;}
		\EndIf
	\end{algorithmic}
\end{algorithm}

\vspace{0.5cm}
Next, the effectiveness of \emph{Algorithm}  $\ref{alg1}$ is quantitatively analyzed.  To this end, we first give the  performance of \emph{Algorithm}  $\ref{alg1}$ without attack.

\begin{proposition}\label{GOpro1}
	Under \emph{Algorithm}  $\ref{alg1}$, if edge  $(j,i)$  is  attack-free,  then  
	\begin{equation*}\label{GOeq222}
		\begin{aligned}
			&D_{KL}(y_{2ij}^{*}(k) \Vert  y_{1ij}^{*}(k))\leqslant \theta. 
		\end{aligned}
	\end{equation*}	 
\end{proposition}

\begin{proof}
	The proof is straightforward  according to  Definition \ref{GOde1}. 
\end{proof}

When  the agent suffers  from attacks, we have the following results.

\begin{Theorem}\label{GOth242}
	Consider the KL divergence $D_{KL}(y_{2ij}^{*}(k) \Vert  y_{1ij}^{*}(k))$ of edge $(j,i)$ with watermarking strategy in Algorithm $\ref{alg1}$. If 
	$\underset{\sigma _{F_1}^{2},\sigma _{F_2}^{2}\rightarrow +\infty}{\lim}\frac{\sigma _{F_1}^{2}}{\sigma _{F_2}^{2}}=0
	$ and $\underset{\lambda_1,\lambda_2\rightarrow +\infty}{\lim}\frac{\lambda_1}{\lambda_2}=0
	$, there exist attacks in  edge $(j,i)$ and  the KL divergence satisfies
	\begin{equation}\label{GOeq99U}
		\begin{aligned}
			\underset{\begin{array}{l}
					\sigma _{F_1}^{2},\sigma _{F_2}^{2}\rightarrow +\infty\\
					\lambda_1,\lambda_2\rightarrow +\infty\\
			\end{array}}{\lim} D_{KL}(y_{1ij}^{*}(k)\Vert y_{2ij}^{*}(k))=+\infty.			
		\end{aligned}
	\end{equation}
\end{Theorem}

\begin{proof}
	For the sake of brevity, the  time index $k$ and  the edge   index  $ij$ for edge $(j,i)$ will be omitted in the sequel. Since the structure of  $y^{*}_1$ and $y^{*}_2$ in one message  set is  the same,  we focus on  $y^{*}_1$   and  replace $1$ with $2$ on the index of  relevant parameters  to generate  the second data in the message  set.
	
	For edge $(j,i)$ suffering  from attacks on the  communication layer,  according to $(\ref{bbp1})$ and 	$(\ref{GOeq22op2})$, the  received data after removing watermarking admits  
	\begin{equation}\label{GOeq100}
		\begin{aligned}
			y_{1}^{*}&=\varXi  _1\left( x+w \right) +M_1\varXi  _1F_1+M_1\varLambda _1-M_1F_1\\
			&=\varXi  _1x+\varXi  _1w+\lambda _1\left( \varXi  _1-\bm{I} \right) F_1+\mathrm{diag}\left\{ M_{1l}^{2} \right\} \varLambda _1\\&~~~+\mathrm{diag}\left\{ M_{1l}^{2} \right\} \left( \varXi  _1-\bm{I} \right) F_1+\lambda _1 \varLambda _1,
		\end{aligned}
	\end{equation}
	and the  $l$-th element of $y_{1}^{*}$  has the form           
	\begin{equation}\label{GOeq1001}
		\begin{aligned}
			y_{1l}^{*}&=\varXi  _{1l}x_l+\varXi  _{1l}w_l+\lambda _{1l}\left( \varXi  _{1l}-1 \right) F_{1l}+M_{1l}^{2}\varLambda _{1l}\\&~~~+M_{1l}^{2}\left( \varXi  _{1l}-1 \right) F_{1l}+\lambda _1\varLambda _{1l}.		
		\end{aligned}
	\end{equation}
	
	We divide  the proof into  8 cases under different kinds of attack parameters, see TABLE $\rm{\ref{GOfi98}}$. 
	Here  $S_1=\left\{ \varXi  _{1l}|\varXi  _{1l}\rightarrow 1 \right\} \bigcup{\left\{ \varXi  _{1l}|\varXi  _{1l}=1 \right\}}$ and $S_2=\left\{ \varXi  _{2l}|\varXi  _{2l}\rightarrow 1 \right\} \bigcup{\left\{ \varXi  _{2l}|\varXi  _{2l}=1 \right\}}$. In this way, $\varXi  _{1r}\notin S_r$ and $\varLambda_{1r}\neq0$ $(r\in \left\{ 1,2 \right\})$ correspond to  multiplicative  and  additive attacks respectively.

\begin{table}[ht]
	\caption{Classification of attack parameters  selection}
	\centering
	\scalebox{0.9}{
		\begin{tblr}{
				cells = {c},
				cell{1}{1} = {c=2,r=2}{},
				cell{1}{3} = {c=2}{},
				cell{1}{5} = {c=2}{},
				cell{3}{1} = {r=2}{},
				cell{3}{3} = {c=2,r=2}{},
				cell{3}{5} = {r=2}{},
				cell{3}{6} = {r=2}{},
				cell{5}{1} = {r=2}{},
				cell{5}{3} = {c=2}{},
				cell{6}{3} = {c=2}{},
				vlines,
				hline{1,3,5,7} = {-}{},
				hline{2} = {3-6}{},
				hline{4} = {2}{},
				hline{6} = {2-6}{},
			}
			{\diagbox[innerwidth=2.8cm]{$y_{2l}^{*}$}{$y_{1l}^{*}$\tnote{1}}} &   & $\varXi  _{1l}\notin S_1$\tnote{1}&   &$\varXi  _{1l}\in S_1 $&   \\
			&   & $\varLambda_{1l}=0$ &$\varLambda_{1l}\neq0$& $\varLambda_{1l}=0$ &$\varLambda_{1l}\neq0$\\
			$\varXi  _{2l}\notin S_2$&$\varLambda_{2l}=0$&	\bm{{Case~1}}&   &	\bm{{Case~2}}&	\bm{{Case~3}}\\
			&$\varLambda_{2l}\neq0$&   &   &     &   \\
			$\varXi  _{2l}\in S_2$&$\varLambda_{2l}=0$&	\bm{{Case~4}}&   &Security &	\bm{{Case~7}}\\
			&$\varLambda_{2l}\neq0$&	\bm{{Case~5}}&   &	\bm{{Case~8}}&	\bm{{Case~6}}
	\end{tblr}}
	\label{GOfi98}
\end{table}

First of all, we will consider the situations where  there exist multiplicative attacks destroying the message set.

	\bm{\mathrm{Case~1:}}	For  convenience, we rewrite  $(\ref{GOeq1001})$ as  {$y_{1l}^{*}=\hat{y}^{*}_{1l}+\hat{y}^{*}_{2l}$}  where {\begin{subequations}\label{testing}
		\begin{align}
			\label{test_1}&\hat{y}^{*}_{1l}=\varXi  _{1l}x_l+\varXi  _{1l}w_l,\\
			\label{test_2}&\hat{y}^{*}_{2l}=\left( M_{1l}^{2}+\lambda _1 \right) \left( \left( \varXi  _{1l}-1 \right) F_{1l}+\varLambda _{1l} \right). 
		\end{align}
\end{subequations}}

{	With  $\eqref{test_2}$,  set $\hat{y}^{*}_{2l}=\hat{y}^{*}_m\hat{y}^{*}_n$ where  $\hat{y}^{*}_m= M_{1l}^{2}+\lambda _1$ and $\hat{y}^{*}_n= \left( \varXi  _{1l}-1 \right) F_{1l}+\varLambda _{1l}$.}

It is obvious that the PDFs of $\hat{y}^{*}_m$ and $\hat{y}^{*}_n$ become

{
	\begin{subequations}\label{211}
		\begin{align}
			\label{212}&f_{\hat{y}^{*}_m}(\gamma )=\begin{cases}
				\frac{\exp \left( -\left( \gamma -\lambda _1 \right) /2\sigma _{M_{1}}^{2} \right)}{\sqrt{2\pi \left( \gamma -\lambda _1 \right) \sigma _{M_{1}}^{2}}},~\gamma >\lambda _1\\
				0,~\gamma \leqslant \lambda _1,\,\,\\
			\end{cases}\\
			&\label{213}f_{\hat{y}^{*}_n}(\gamma )=\frac{\exp \left( -\left( \gamma -\varLambda _{1l} \right) ^2/2\left( \varXi  _{1l}-1 \right) ^2\sigma _{F_{1}}^{2} \right) }{\sqrt{2\pi \left( \varXi  _{1l}-1 \right) ^2\sigma _{F_{1}}^{2}}}.
		\end{align}
\end{subequations}}

In terms of $(\ref{212})$ and  $(\ref{213})$, the PDF  of  { $\hat{y}^{*}_{2l}$} is 
 {	\begin{equation}\label{223}
		\begin{aligned}
			f_{\hat{y}^{*}_{2l}}(\gamma )&=\int_{\lambda _1}^{+\infty}{\frac{1}{\left| z \right|}}f_{\hat{y}^{*}_m}(z )f_{\hat{y}^{*}_n}(\frac{\gamma}{z})dz\\
			&\xlongequal{\left( a \right)}
			\int_{\lambda _1}^{+\infty}g_{\hat{y}^{*}_{2l}}(z)dz,
		\end{aligned}
\end{equation}
where $(a)$ holds for  $\lambda_{1} \in \left( 0,+\infty \right)$ and $g_{\hat{y}^{*}_{2l}}(z)=			
	{\frac{1}{z}}f_{\hat{y}^{*}_m}(z )f_{\hat{y}^{*}_n}(\frac{\gamma}{z})$}.

Based on $(\ref{213})$, it turns out that 
 {
	\begin{equation}\label{225ss}
		\begin{aligned}
			&f_{\hat{y}^{*}_n}\left( \gamma \right)=\int_{\lambda _1}^{+\infty}{g_{{\hat{y}^{*}_n}}(z)}dz,
		\end{aligned}
	\end{equation}
	where
	\begin{equation*}\label{225osd}
		\begin{aligned}
			g_{{\hat{y}^{*}_n}}(z)=&\exp \left( -\ln ^2\left( z-\lambda _1 \right) +2\ln \left( z-\lambda _1 \right) \right)\frac{f_{\hat{y}^{*}_n}(\gamma)}{e\sqrt{\pi}\left( z-\lambda _1 \right)}.
		\end{aligned}
\end{equation*}}

From $(\ref{223})$ and $(\ref{225ss})$, it is straightforward that there exists a  positive  constant $\hat{\xi}_{1_1}$ which is  bounded away from zero, such that $\frac{1}{\hat{\xi}_{1_1}}f_{\hat{y}^{*}_n}\left( \gamma \right) <\,\,f_{y_{2l\,\,}}\left( \gamma \right)$. In this way, the PDF of $y_{1l}^{*}$ can be scaled as  
 { 
	\begin{equation}\label{225sd}
		\begin{aligned}
			f_{y_{1l}^{*}}\left( \gamma \right) &=\int_{-\infty}^{+\infty}{f_{\hat{y}^{*}_{1l\,\,}}\left( z \right) f_{\hat{y}^{*}_{2l\,\,}}\left( \gamma -z \right)}dz
			\\&<\frac{1}{\hat{\xi}_{1_1}}\int_{-\infty}^{+\infty}{f_{\hat{y}^{*}_{n\,\,}}\left( z \right) f_{\hat{y}^{*}_{2l\,\,}}\left( \gamma -z \right)}dz
			\\&=\frac{\exp \left( -\frac{\left( \gamma -\left( \varLambda _{1l}+\varXi  _{1l}x_l \right) \right) ^2}{2\left( \left( \varXi  _{1l}-1 \right) ^2\sigma _{F_1}^{2}+\varXi  _{1l}^{2}\sigma _{1}^{2} \right)} \right)}{\hat{\xi}_{1_1}\sqrt{2\pi \left( \left( \varXi  _{1l}-1 \right) ^2\sigma _{F_1}^{2}+\varXi  _{1l}^{2}\sigma _{1}^{2} \right)}}.
		\end{aligned}
\end{equation}}
Similarly, the PDF of  $y_{2l}^{*}$ can be transformed by inequality manipulation  into
 { 
	\begin{equation}\label{225sf}
		\begin{aligned}
			f_{y_{2l}^{*}}\left( \gamma \right) >\frac{\exp \left( -\frac{\left( \gamma -\left( \varLambda _{2l}+\varXi  _{2l}x_l \right) \right) ^2}{2\left( \left( \varXi  _{2l}-1 \right) ^2\sigma _{F_2}^{2}+\varXi  _{2l}^{2}\sigma _{1}^{2} \right)} \right)}{\hat{\xi}_{1_2}\sqrt{2\pi \left( \left( \varXi  _{2l}-1 \right) ^2\sigma _{F_2}^{2}+\varXi  _{2l}^{2}\sigma _{1}^{2} \right)}}.
		\end{aligned}
\end{equation}}

From $(\ref{225sd})$ and $(\ref{225sf})$, it yields 
\begin{equation*}\label{225shy}
	\begin{aligned}
		&D_{KL}(y_{1ij}^{*}\Vert y_{2ij}^{*})
		\\>&\frac{1}{2\hat{\xi}_{1_1}}\sum_{l=1}^n\left( \log \frac{\left( \varXi  _{2l}-1 \right) ^2\sigma _{F_2}^{2}+\varXi  _{2l}^{2}\sigma _{1}^{2}}{\left( \varXi  _{1l}-1 \right) ^2\sigma _{F_1}^{2}+\varXi  _{1l}^{2}\sigma _{1}^{2}}-1
		\right.\\ &\left.+\frac{\left( \varXi  _{1l}-1 \right) ^2\sigma _{F_1}^{2}+\varXi  _{1l}^{2}\sigma _{1}^{2}}{\left( \varXi  _{2l}-1 \right) ^2\sigma _{F_2}^{2}+\varXi  _{2l}^{2}\sigma _{1}^{2}}
		\right.\\ &\left.+\frac{\left( \left( \varLambda _{2l}+\varXi  _{2l}x_l \right) -\left( \varLambda _{1l}+\varXi  _{1l}x_l \right) \right) ^2}{\left( \varXi  _{2l}-1 \right) ^2\sigma _{F_2}^{2}+\varXi  _{2l}^{2}\sigma _{1}^{2}} \right)+\frac{n}{\hat{\xi}_{1_1}}\log \frac{\hat{\xi}_{1_2}}{\hat{\xi}_{1_1}}.
	\end{aligned}
\end{equation*}

Now we  can see that $\varXi  _{1l}$, $\varXi  _{2l}$, $\hat{\xi}_{1_h}$ and  $\sigma _{1}^{2}$  with $h\in \left\{ 1,2 \right\} $ are bounded.  Therefore, if   $\underset{\sigma _{F_1}^{2},\sigma _{F_2}^{2}\rightarrow +\infty}{\lim}\frac{\sigma _{F_1}^{2}}{\sigma _{F_2}^{2}}=0
$, we have 

\begin{equation}\label{2225sf}
	\begin{aligned}
		\underset{\sigma _{F_1}^{2},\sigma _{F_2}^{2}\rightarrow +\infty}{\lim} D_{KL}(y_{1ij}^{*}||y_{2ij}^{*})=+\infty.
	\end{aligned}
\end{equation}

\bm{{Case~2:}} In this condition, only  $\widetilde{y}_{2}$ in the message set  transmitted suffers from the attack such that  $y_{1l}^{*}=x_{l}+w_{1l}$, where the PDF is 
 {
	$$
	f_{y_{1l}^{*}}(\gamma )=\frac{\exp \left( -\left( \gamma -x_l \right) ^2/2\sigma _{1}^{2} \right)}{\sqrt{2\pi \sigma _{1}^{2}}}.
	$$}

By following a similar procedure
of \bm{{Case~1}}, it indicates 

\begin{equation}\label{225sgh}
	\begin{aligned}
		&D_{KL}(y_{1ij}^{*}\Vert y_{2ij}^{*})
		\\>&\frac{1}{2}\sum_{l=1}^n\left( \log \frac{\left( \varXi  _{2l}-1 \right) ^2\sigma _{F_2}^{2}+\varXi  _{2l}^{2}\sigma _{1}^{2}}{\sigma _{1}^{2}}-1
		\right.\\ &\left.+\frac{\sigma _{1}^{2}}{\left( \varXi  _{2l}-1 \right) ^2\sigma _{F_2}^{2}+\varXi  _{2l}^{2}\sigma _{1}^{2}}
		\right.\\ &\left.+\frac{\left( \left( \varLambda _{2l}+\varXi  _{2l}x_l \right) -x_l \right) ^2}{\left( \varXi  _{2l}-1 \right) ^2\sigma _{F_2}^{2}+\varXi  _{2l}^{2}\sigma _{1}^{2}} \right)+n\log \hat{\xi}_{1_2}.
	\end{aligned}
\end{equation}

Based on $(\ref{225sgh})$, since $\varXi  _{2l}$ and  $\hat{\xi}_{1_2}$  are bounded, we have
\begin{equation}\label{2w5sf}
	\begin{aligned}
		\underset{\sigma _{F_2}^{2}\rightarrow +\infty}{\lim} D_{KL}(y_{1ij}^{*}(k)||y_{2ij}^{*}(k))=+\infty.
	\end{aligned}
\end{equation}

\bm{{Case~3:}} In this scenario, $\widetilde{y}_{1}$ only suffers from the  additive attack and  {$y_{1l}^{*}=\hat{y}^{*}_{a}+\hat{y}^{*}_{b}$
	, where $\hat{y}^{*}_{a}=x_l+w_{l}+\lambda _1\varLambda _{1l}$ and $\hat{y}^{*}_{b}= M_{1l}^{2}\varLambda _{1l}$. If $\varLambda _{1l}>0$, the  PDFs  of $\hat{y}^{*}_{a}$ and $\hat{y}^{*}_{b}$ are
 { 
		\begin{subequations}\label{2aae}
			\begin{align}
				\label{2aw1}&f_{\hat{y}^{*}_a}(\gamma )=\frac{\exp \left( -\left( \gamma -x_l-\lambda _1\varLambda _{1l} \right) ^2/2\sigma _{1}^{2} \right)}{\sqrt{2\pi \sigma _{1}^{2}}},
				\\
				&\label{2aw2}f_{\hat{y}^{*}_b}(\gamma )=\begin{cases}
					\frac{\exp \left( -\gamma /2\sigma _{M_1}^{2}\varLambda _{1l} \right)}{\sqrt{2\pi \gamma \varLambda _{1l}\sigma _{M_1}^{2}}},~\gamma >0\\
					0,~ \gamma \leqslant 0.\,\,\\
				\end{cases}
			\end{align}
\end{subequations}}}

According to $(\ref{2aw1})$ and $(\ref{2aw2})$, the PDF of $y_{1l}^{*}$  is 
 { 
	\begin{equation*}\label{2s25}
		\begin{aligned}
			f_{y_{1l}^{*}}(\gamma )=\int_{0}^{+\infty}{\frac{\exp \left( -\frac{z}{2\sigma _{M_1}^{2}\varLambda _{1l}}-\frac{\left( \gamma -z-x_l-\lambda _1\varLambda _{1l} \right) ^2}{2\sigma _{1}^{2}} \right)}{\sqrt{4\pi ^2\sigma _{1}^{2}\varLambda _{1l}\sigma _{M_1}^{2}\sigma _{1}^{2}z}}dz}.
		\end{aligned}
\end{equation*}}

$(\ref{2aw1})$  can be further transformed into 
 {\begin{equation}\label{22ess}
		\begin{aligned}
			&f_{\hat{y}^{*}_a}\left( \gamma \right)=\int_{0}^{+\infty}{g_{{\hat{y}^{*}_a}}(z)}dz,
		\end{aligned}
	\end{equation}
	where
	\begin{equation*}\label{29sd}
		\begin{aligned}
			g_{{\hat{y}^{*}_a}}(z)=\exp \left( -\ln ^2\left( z\right) +2\ln \left( z\right) \right)\frac{f_{\hat{y}^{*}_a}(\gamma)}{e\sqrt{\pi}\left( z\right)}.
		\end{aligned}
\end{equation*}}

Similar to $(\ref{225sd})$, there exists a  positive  constant $\hat{\xi}_{3_1}$ which is  bounded away from zero, such that $\frac{1}{\hat{\xi}_{3_1}}f_{\hat{y}^{*}_n}\left( \gamma \right) <\,\,f_{y_{1l}^{*}}\left( \gamma \right)$. Then,  
\begin{equation*}\label{225sh2y}
	\begin{aligned}
		&D_{KL}(y_{1ij}^{*}\Vert y_{2ij}^{*})
		\\>&\frac{1}{2\hat{\xi}_{3_1}}\sum_{l=1}^n\left( \log \frac{\left( \varXi  _{2l}-1 \right) ^2\sigma _{F_2}^{2}+\varXi  _{2l}^{2}\sigma _{1}^{2}}{\sigma _{1}^{2}}-1
		\right.\\ &\left.+\frac{\sigma _{1}^{2}}{\left( \varXi  _{2l}-1 \right) ^2\sigma _{F_2}^{2}+\varXi  _{2l}^{2}\sigma _{1}^{2}}
		\right.\\ &\left.+\frac{\left( \left( \varLambda _{2l}+\varXi  _{2l}x_l \right) -\left( \varLambda _{1l}+x_l \right) \right) ^2}{\left( \varXi  _{2l}-1 \right) ^2\sigma _{F_2}^{2}+\varXi  _{2l}^{2}\sigma _{1}^{2}} \right)+\frac{n}{\hat{\xi}_{3_1}}\log \frac{\hat{\xi}_{1_2}}{\hat{\xi}_{3_1}}.
	\end{aligned}
\end{equation*}

Performing the same procedure  as \bm{{Case~2}}, one gets  
\begin{equation}\label{235sf}
	\begin{aligned}
		\underset{\sigma _{F_2}^{2}\rightarrow +\infty}{\lim} D_{KL}(y_{1ij}^{*}||y_{2ij}^{*})=+\infty.
	\end{aligned}
\end{equation}

Since the treatment of scenario  $\varLambda _{1l}<0$ is similar to the one that $\varLambda _{1l}>0$, no relevant research will be further conducted.

For \bm{{Case~4}} and \bm{{Case~5}}, through the same analytical framework in \bm{{Case~2}} and \bm{{Case~3}}, it follows that 
\begin{equation}\label{23533f}
	\begin{aligned}
		\underset{\sigma _{F_1}^{2}\rightarrow +\infty}{\lim} D_{KL}(y_{1ij}^{*}||y_{2ij}^{*})=+\infty.
	\end{aligned}
\end{equation}

Next, we will address  the situations that the message set only suffers from the additive attack.

\bm{{Case~6:}} In this scenario, we have 
\begin{equation*}\label{GOeq1w0}
	\begin{aligned}
		y_{1l}^{*}=x_l+w_l+M_{1l}^{2}\varLambda _{1l}+\lambda_1\varLambda _{1l},\\
		y_{2l}^{*}=x_l+w_l+M_{2l}^{2}\varLambda _{2l}+\lambda_2\varLambda _{2l}.		
	\end{aligned}
\end{equation*}

Spurred by the arguments in \bm{{Case~2}} and \bm{{Case~3}}, it is indicated  that  if   $\underset{\lambda_1,\lambda_2\rightarrow +\infty}{\lim}\frac{\lambda_1}{\lambda_2}=0
$, we have 
\begin{equation}\label{235223f}
	\begin{aligned}
		\underset{\lambda_1,\lambda_2\rightarrow +\infty}{\lim} D_{KL}(y_{1ij}^{*}||y_{2ij}^{*})=+\infty.
	\end{aligned}
\end{equation}

As for \bm{{Case~7}} and \bm{{Case~8}}, the above conclusion can also be drawn out through similar technical  routes in \bm{{Case~6}} and will not be repeated here. 

Based on the statements  discussed above, the proof is thus completed.
\end{proof}

\vspace{0.7cm}
\begin{rmk}
	According to the proof of Theorem $\ref{GOth242}$, we can see that additive attacks are subject to multiplicative watermarking while multiplicative attacks will be affected by  additive watermarking. In  \emph{Algorithm}  $\ref{alg1}$, two types of watermarking are involved, allowing for simultaneous disclosure of additive and multiplicative attacks.
\end{rmk}
\vspace{0.7cm}
\begin{rmk}
	In contrast to  \cite{9146361} and  \cite{mustafa2020resilient}, the proposed scheme is capable of  transient stage since the analysis process spans the entire  system evolution time. In other words,  it  is  not only   applicable to the situation  $k\rightarrow \infty$. The applicability of transient stage  comes from  our special provision that each agent  transmits a message set  on the communication layer, rather than a single state value  under traditional communication strategies. At the same time, in order to prevent attackers from eavesdropping on the detection strategy, we match different watermarking  strategies for different state values in one  message set. This lays the foundation for real-time data verification subsequently.
\end{rmk}

 {
	\begin{rmk}
		As shown in the proof of Theorem $\ref{GOth242}$, for the attacks containing  multiplicative signals which correspond to Cases 1-5, the additive watermarking with parameters  $\sigma _{F_l}^{2}, l\in \left\{ 1,2 \right\} $ plays a major  role. See $(\ref{2225sf})$ in Case 1, $(\ref{2w5sf})$ in  Case 2, $(\ref{235sf})$ in  Case 3 and $(\ref{23533f})$ in  Cases  4 and 5.
		For  Cases 6-8 having single  additive attacks, it turns  to  multiplicative watermarking with relevant parameters as  ${\lambda_l}, l\in \left\{ 1,2 \right\} $ to keep  a central position in attack detection. Therefore, the parameters involving  both   multiplicative and additive watermarking are settled in $(\ref{GOeq99U})$ in Theorem $\ref{GOth242}$.
		Besides, once an attack occurs,  the inverse process under the watermarking  removal mechanism at the receiving  end becomes invalid. This leads to the presence of watermarking   relevant statistical characteristics  in the attacked data in the message set $\left\{ \overline{y}^{a}_{1ij}\left( k \right) ,\overline{y}^{a}_{2ij}\left( k \right) \right\}$ at the receiving end, which   is the key  factor that exposes an attack on the communication layer.
\end{rmk}}

 {
\begin{rmk}
There are currently  numerous  model-based fault/attack detection mechanisms available.   In \cite{zhang2010fault}, an adaptive fault diagnosis method was   proposed by synthesizing fault detection estimators and a bank of fault isolation estimators. A detector based on the  Luenberger observer together with a series of   unknown input observers was  designed to implement distributed detection in DC microgrids in \cite{gallo2020distributed}. In \cite{barboni2020detection}, a detection mechanism with the foundation of   model-based observer  was   put forward for covert attacks in interconnected systems to implement distributed attack detection among subsystems. 
However, a common feature of the aforementioned work is that they  are all  built on the observer-based  framework. That is to say, the successful defense against attacks  strongly depends on the  convergence of estimation error.  This also leads to its inapplicability for   transient stage.  
\end{rmk}}

\subsection{Detection  Strategy for Attacks on the Agent layer}\label{opp2}

In recent years, the attack  issue  of  agent layer has also received much  attention  \cite{leblanc2013resilient,wu2020federated}. Typical attacks launched include malicious attacks and Byzantine attacks.
In this subsection, we focus on  the  latter  which is implemented  by  sending  malicious information to different out-neighbors  \cite{ishii2022overview}. Thus it is  a  more  threatening one than malicious attack which spreads influence through  broadcasting of  malicious agents. Besides, we pay attention to the Byzantine attacks  that  shake the convergence of MASs.

The existing literature  shows that the norm of tracking error between the leader and followers  without attack has an envelope of the upper bound \cite{cheng2016convergence}. That is, if the   system operates normally, the norm of tracking error will always be restricted  by an  envelope.  In this paper, we choose the envelope with the form 
$$\tau \left( k \right)=\varrho e^{ -g(k)},$$ 
where $g(k)\in \mathbb{R}$ is a positive increasing function with  time $k$. Moreover,  the consensus  problem in  MASs is usually recast into the stability issue  of error system. To this end, we define the  error signal  $d_{ij}(k)\triangleq \mathbb{E}\left\| y_{ij}\left( k \right)-x_{i}\left( k \right)\right\|$. In this  section, in  order to concentrate defensive forces on attacks that affect system performance and avoid unnecessary losses,  we investigate  a class of attackers that impact the convergence performance which manifests as  overstepping  the  envelope of state error in normal systems. For convenience, the Byzantine attacks mentioned below refer to the ones  that the state error exceeds the envelope.

If $d_{ij}(k) <d_{ij}(k-1) \left( \tau \left( k \right) +\delta \right)$,  agent  $j$ is free of  Byzantine attack; otherwise,   $j$  is appointed to be a  Byzantine agent.  Here    $\delta \in \mathbb{R}$ is an  offset  indicating the resilience of system. \emph{Algorithm} $\ref{alg2}$ illustrates  the detailed steps for Byzantine attack detection strategy and  its  skeleton is shown as below.

\begin{figure}[H]
	\centering
	\includegraphics[width=3.45in,height=0.81in]{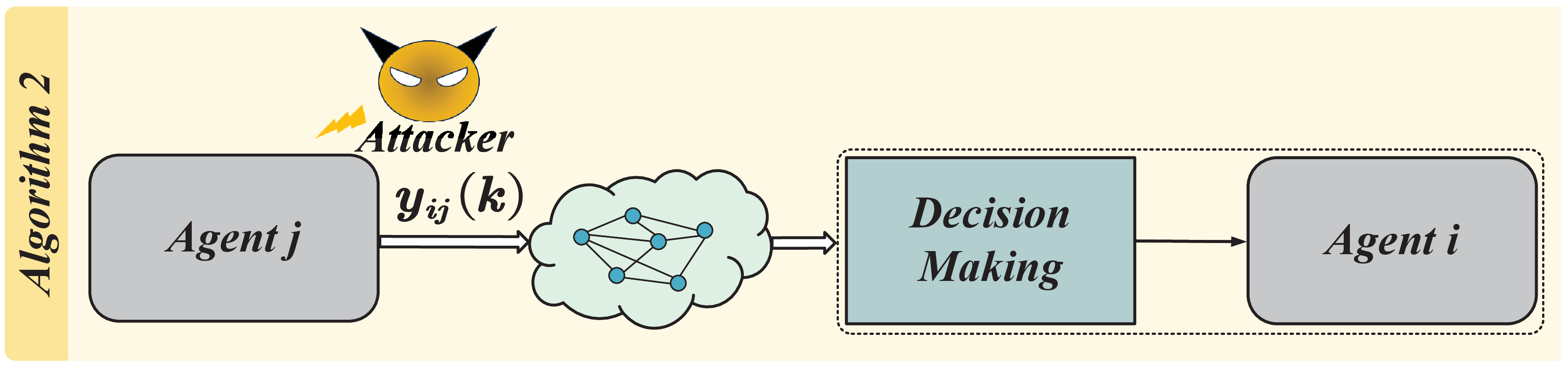}
	\caption{Block diagram of  \emph{Algorithm}  $
		\ref{alg2}$ for  attacks  on the  agent  layer.}
	\label{GOfig142}\end{figure}

\begin{algorithm}[H]
	\caption{Byzantine Attacks Detection Strategy}
	\label{alg2}
	\begin{algorithmic}[1]
		\State{Each agent $j$ with $j\in{N_{i}^{+}}$ sends the message  $y_{ij}(k)$}  into communication edge $(j,i)$;
		\State{Decision making: Calculate $d_{ij}\left( k \right)$ and $\tau(k)$;} 
		\If {$\mathbb{E}\left(d_{ij}\left( k \right)\right)  \leqslant\sqrt{\frac{\epsilon_{1}^{2}+\epsilon_{2}^{2}}{\epsilon _{2}^{2}}} \mathbb{E}\left(d_{ij}\left( k-1 \right)\right)  \left( \tau \left( k \right) +\delta \right)$}
		\State {Agent $j$ is not attacked;}
		\Else {~$\mathbb{E}\left(d_{ij}\left( k \right)\right) >\sqrt{\frac{\epsilon _{1}^{2}+\epsilon _{2}^{2}}{\epsilon _{2}^{2}}}\mathbb{E}\left(  d_{ij}\left( k-1 \right) \right) \left( \tau \left( k \right) +\delta \right) $}		
		\State { $j$ is a  Byzantine agent;}
		\EndIf
	\end{algorithmic}
\end{algorithm}

Before demonstrating  the effectiveness of  \emph{Algorithm}  $\ref{alg2}$, we first introduce  Lemma $\ref{GOle1}$ to assist the   subsequent analysis. 

\begin{Lemma}\label{GOle1}
	For  vectors $\varGamma \in \mathbb{R}^{n}$ and $\varOmega \in \mathbb{R}^{n}$ in which $\varGamma_{i}, \varOmega_{i}\in \left[ \varrho _1,\varrho _2 \right]$ with $ i\in \left\{ 1,\cdots ,n \right\} $, and  $\varGamma_{i}$ and $\varOmega_{i}$ are the $i$-th element of $\varGamma$ and $\varOmega$, we have
 {
		\begin{equation}\label{GOeq1929}
			\begin{aligned}
				\left\| \varGamma \right\| +\left\| \varOmega \right\| \leqslant\sqrt{\frac{\varrho _{1}^{2}+\varrho _{2}^{2}}{\varrho _{1}^{2}}}\left\| \varOmega +\varGamma \right\| .
			\end{aligned}
	\end{equation}}
\end{Lemma}

\begin{proof}
	According to the definition of 	2-norm, it is clear
	 {
		\begin{equation*}\label{Geeq1929}
			\begin{aligned}
				&\left( \frac{\left\| \varOmega +\varGamma \right\|}{\left\| \varGamma \right\| +\left\| \varOmega \right\|} \right) ^2
				\\=&\frac{\sum\limits_{i=1}^n{\left( \varOmega _{i}^{2}+\varGamma _{i}^{2} \right)}+\sum_{i=1}^n\limits{2\varOmega _i\varGamma _i}}{\sum_{i=1}^n\limits{\left( \varOmega _{i}^{2}+\varGamma _{i}^{2} \right)}+2\sqrt{\sum_{i=1}^n\limits{\varOmega _{i}^{2}}\sum_{i=1}^n\limits  {\varGamma _{i}^{2}}}}
				\\\overset{\text{(d)}}{>}&\frac{\varrho _{1}^{2}}{\varrho _{1}^{2}+\varrho _{2}^{2}}.
			\end{aligned}
	\end{equation*}}
 {	Specially, the  inequality manipulation  $(d)$ comes from 1) $2n\varrho _{1}^{2}\leqslant \sum\limits_{i=1}^n{\left( \varOmega _{i}^{2}+\varGamma _{i}^{2} \right)}\leqslant 2n\varrho _{2}^{2}$; 2) $\sum\limits_{i=1}^n{2\varOmega _i\varGamma _i}\geqslant 0$; 3) $2\sqrt{\sum\limits_{i=1}^n{\varOmega _{i}^{2}}\sum\limits_{i=1}^n{\varGamma _{i}^{2}}}\leqslant 2n\varrho _{1}^{2}$. Besides, since each scaling from $1)$ to $3)$   is  fulfilled  independently but cannot be met simultaneously,  ``='' will not occur. This completes the proof. }     
\end{proof}

For the convenience of  analysis,  an assumption is provided.

\begin{ass}\label{GOas29} \hspace{-0.001cm} {\cite{cheng2016convergence}}
	
	(A4-1) 	There exist positive constants $\mu _1\leqslant \mu _2<\infty$ and $\varLambda \in \left( 0,1 \right)$ such that for any $k\geqslant 1$, $\mu _1k^{-\varLambda}\leqslant a_i\left( k \right) \leqslant \mu _2k^{-\varLambda}$;
	
	(A4-2) 	$\lim\limits_{k\rightarrow \infty}a_i\left[ k \right] /a_j\left[ k \right] =c_{ij}$, $c_{ij}>0$.
\end{ass}

Then, for any normal agent $i$ and its normal neighbor $j \in N_i^{+}$, we give  the following distributed condition about the convergence rate.

\begin{proposition}\label{GOpro22}
	Given  $\tau \left( k \right) =M_re^{\left( -\lambda _{\min}k^{\left( 1-\phi \right)} \right)}$, where $M_r$ is a bounded positive constant,    $\lambda _{\min}$ is the smallest eigenvalue  of $CL_2$ with  $C=\mathrm{diag}\left\{ c_1,...,c_N \right\} 
	$, $\underset{t\rightarrow \infty}{\lim}\frac{a_i\left( t \right)}{\overline{a}\left( t \right)}=c_i\,\,\left( i=1,...,N \right) 
	$,  $\overline{a}\left( t \right) =\underset{i=1,...,N}{\max}\left\{ a_i\left( t \right) \right\} 
	$ and $\phi \in (0,\operatorname*{min}\{1,\lambda_{\mathrm{min}}\})$, suppose \emph{Assumption} $\ref{GOas29}$ is satisfied. Then, for any  normal agent $i$, it follows that 
	\begin{equation}\label{GOeq992}
		\begin{aligned}
			d_{ij}(k)\leqslant\sqrt{\frac{\epsilon _{2}^{2}}{\epsilon_{1}^{2}+\epsilon_{2}^{2}}}\left( \tau \left( k \right) +\delta \right)d_{ij}(k-1),~j\in N_i^{+}.
		\end{aligned}
	\end{equation}			
\end{proposition}

\begin{proof}
	According to \cite{cheng2016convergence}, by \emph{Assumption} $\ref{GOas1}$, one gets   $d_{i0}\left( k \right) <d_{i0} \left( k-1 \right) \left( \tau \left( k \right) +\delta \right) $ when there is no attack.

	Define a path from leader $0$ to follower $i$ be $\left( v_0,v_{i_1},...,v_{i_m},...,v_{i_{p-2}},v_i \right)$, then  
	
	 {
		\begin{equation}\label{GOeq1962}
			\begin{aligned}
				&\mathbb{E}\left\| x_0\left( k \right) -x_i\left( k \right) \right\| 
				\\=~&\mathbb{E}\left\| x_0\left( k \right) -x_{i_1}\left( k \right) +x_{i_1}\left( k \right) +...
				\right. \\&
				\left.+x_{i_{m-1}}(k)-x_{i_m}\left( k \right) +...+x_{i_{p-2}}-x_i\left( k \right) \right\| 
				\\=~&\mathbb{E}\left\|\widetilde{y}_{i_10}\left( k \right) -x_{i_1}\left( k \right) +\widetilde{y}_{i_2i_1}\left( k \right) -x_{i_2}\left( k \right)+...\right. \\&
				\left.+ \widetilde{y}_{i_mi_{m-1}}(k)-x_{i_m}\left( k \right) +...+\widetilde{y}_{ii_{p-2}}-x_i\left( k \right) \right\| 
				\\ \leqslant~&\mathbb{E}\left\| \widetilde{y}_{i_10}\left( k \right) -x_{i_1}\left( k \right) \right\| +\mathbb{E}\left\| \widetilde{y}_{i_2i_1}\left( k \right) -x_{i_2}\left( k \right) \right\|+...
				\\& +\mathbb{E}\left\| \widetilde{y}_{i_mi_{m-1}}(k)-x_{i_m}\left( k \right) \right\|+...
				+\mathbb{E}\left\| \widetilde{y}_{ii_{p-2}}(k)-x_i(k) \right\|. 
			\end{aligned}
	\end{equation}}
	
	From Lemma $\ref{GOle1}$ and $(\ref{GOeq1962})$, if for $\forall~m\in \left\{ 1,\cdots ,p-2 \right\}$,
	 {
		\begin{equation}\label{GOeq192}
			\begin{aligned}
				&\mathbb{E}\left\| \widetilde{y}_{i_mi_{m-1}}(k)-x_{i_m}\left( k \right) \right\|
				\\\leqslant&\sqrt{\frac{\epsilon _{2}^{2}}{\epsilon_{1}^{2}+\epsilon _{2}^{2}}}\left( \tau \left( k \right) +\delta \right)\mathbb{E}\left\| \widetilde{y}_{i_mi_{m-1}}(k-1)-x_{i_m}\left( k-1 \right) \right\|,
			\end{aligned}
	\end{equation}}
	we have
 {
		\begin{equation}\label{GOeq192}
			\begin{aligned}
				\mathbb{E}\left\| x_0\left( k \right) -x_i\left( k \right) \right\|
				\leqslant \left( \tau \left( k \right) +\delta \right)\mathbb{E} \left\| x_0\left( k-1 \right) -x_i\left( k-1 \right) \right\|,
			\end{aligned}
	\end{equation}}
	which means that $\left\| x_0\left( k \right)-x_i\left( k \right) \right\|$ will converge to  zero over time. The proof is completed. 
\end{proof}

Therefore, under \emph{Algorithm} $\ref{alg2}$, if  there is no attack, the detector will not be  alarmed. While if  the system suffering from  Byzantine attacks  causes   $\left\| x_{i,{p-2}}\left( k \right)-x_i\left( k \right) \right\|>\sqrt{\frac{\varrho _{1}^{2}+\varrho _{2}^{2}}{\varrho _{2}^{2}}}\left( \tau \left( k \right) +\delta \right) \left(x_{i,{p-2}}\left( k-1 \right)-x_i\left( k-1 \right)\right)$, the detector alarm is  triggered.

Now, the effectiveness of \emph{Algorithm} $\ref{alg2}$ is  portrayed by  the following proposition.

\begin{proposition}\label{GOpro2}
	If $\tau \left( k \right) $ is chosen as  $M_re^{\left( -\lambda _{\min}k^{\left( 1-\phi \right)} \right)}$ and \emph{Assumption} $\ref{GOas29}$ is guaranteed,  \emph{Algorithm} $\ref{alg2}$ can detect  Byzantine attacks of which the  state value  exceeds the state envelope that matches  the convergence rate.
\end{proposition}

\begin{proof}
	Based on Proposition $\ref{GOpro22}$ and the  types of target attacks to be detected, the proof can be directly obtained.
\end{proof}

\begin{rmk}
	Proposition \ref{GOpro2} provides a solution  to examine whether there exists a Byzantine attack or not. Since the detection scheme is related to the   convergence rate,  it is also valid for  transient stage. Unlike  strategy in  \cite{zhou2022watermarking}, the proposed scheme concentrates on the impact of attack  on convergence. Such a treatment  occupies less  resources while guaranteeing a desirable consensus performance. Besides, $M_r$ and   $\delta$   reflect the level of  resilience  to attacks.
\end{rmk}

 {\begin{rmk}
		There are some  options for $\tau(k)$, such as $\varrho e^{ -g(k)}$ \cite{cheng2016convergence} and $c\varrho ^k$ \cite{yan2022resilient}. 		
		The design idea of the detector in this section is to provide  an envelope of decreasing time in terms of  the upper bound of the relative error of state values among agents, so as to conform to the operation law of the normal system. In this way, the malicious agents  that affect the system convergence are screened out.		
		However, from the theoretical perspective, we aim  to  find  a detector where its  effectiveness can be rigorously demonstrated.
		With  the leader-follower multi-agent model in  $(\ref{GOeq1})$, the upper bound of the convergence rate for the  tracking error between the  leader and the  follower  is $M_re^{\left( -\lambda _{\min}k^{\left( 1-\phi \right)} \right)}$.  However, the envelope of the  detector  cannot  be directly adopted for this kind of $\tau(k)$. Considering the fully distributed requirement of the algorithm proposed in this paper, not every agent  can acquire  the state of the leader.
		That is to say, the results of this article, namely Propositions  3 and  4, cannot be directly obtained from the relevant  literature   in \cite{cheng2016convergence,7554644}. Therefore, it is necessary to transform the detection indicators from differences in the state of followers and the leader to differences in the state of neighbors, in order to achieve full distribution of the algorithm.
\end{rmk}}

 {\begin{rmk}
		$\delta$ actually reflects the degree of   the envelope $\tau(k)$ moving  up along the  vertical axis, which is related to the robustness of the detector. The larger the envelope of $\tau(k)$ is, the greater  tolerance for outliers   will be with  lower false alarm rate to the corresponding  attacks. Therefore,  this parameter is very important, especially for  $(\ref{GOeq1})$ with uncertain factors, which is more realistic. However, an inappropriate   high  value of $\delta$  will cause  the failure to detect malicious attacks on the target, ultimately leading to a decrease in detection rate. This understanding has been clearly stated  in the design of  detectors, for example, \cite{75MO,8727926}.
\end{rmk}}

 {\begin{rmk}		   
	A common sense needs to be stated here: when the system needs more agents  to be isolated,  more redundant edges are required in the initial graph to maintain the connectivity of the network communication topology which implies a   higher  connectivity of the initial graph  \cite{ishii2022overview}.
	Next, two  cases are considered:	
	a) Too low detection accuracy: This can be  characterized  in extreme cases. That is, the detector cannot detect misbehaving  nodes that have too much influence on the system, so that normal agents  maintain information interaction with these extremely misbehaving  agents. It may  cause  malicious information to be injected, preventing the realization of resilient     consensus, and greatly reducing  system resilience.
	b) Too high  detection accuracy: Here we give  an example  when a malicious message sent by a misbehaving agent is small enough, even combined with a short attack time, that will  not  disrupt the  convergence.  If the detector has high accuracy and exposes this attack, it leads to a series of node isolation behaviors \cite{yuan2021secure,mustafa2020resilient}. In addition, under the same number of Byzantine agents, more redundant edges are needed to  maintain the connectivity of the communication topology.
	Naturally,    a high connectivity  of the initial graph will  weaken  system resilience.  
	Thus, we consider the  issue of balancing  system resilience and  detection accuracy in this paper. If the attack  does  not affect the convergence of the system, the detector is expected to tolerate these malicious behaviors.	
	To achieve this goal, we design an envelope-based detector to identify whether the behavior of misbehaving  agents will disrupt the convergence trend of the system.  
	In particular, the envelope  is a decreasing curve that characterizes the upper bound on the tracking error of the system. 	In this way, only the  attacked  agents  that affect the convergence trend of the system are identified, and the  malicious behaviors that do not affect the convergence trend of the system will be  ignored. By doing so,  the secure resilience is enhanced. This implements   the tradeoff   what we  call.
\end{rmk}}

\vspace{0.5cm}

\subsection{Detection  Strategy for Hybrid Attacks}\label{opp3}

In this  subsection, we design  the detection scheme for  hybrid attacks that  can  destroy  communication channel and agent at the same time. It is noted that   most   existing  literature on  detection strategy becomes    infeasible for the problem formulated here. This is because the behavior of hybrid attacks becomes  more complicated and the overall effect is not a simple superposition of two attacks. 
For example, a  detector may  regard the  attacks on the communication  layer  as Byzantine attacks. This will cause a  degradation of control performance and  increase the security cost. In addition, a simple parallel or composition  of two detection strategies will become ineffective.  Hence, the study on hybrid attacks  is somewhat challenging and meaningful. A framework  of detection strategy, e.g., \emph{Algorithm} $\ref{alg3}$ in this  subsection is shown in Fig.
$\ref{GOfig66}$.

\begin{figure}[H]
	\centering
	\includegraphics[width=3.46in,height=3.3in]{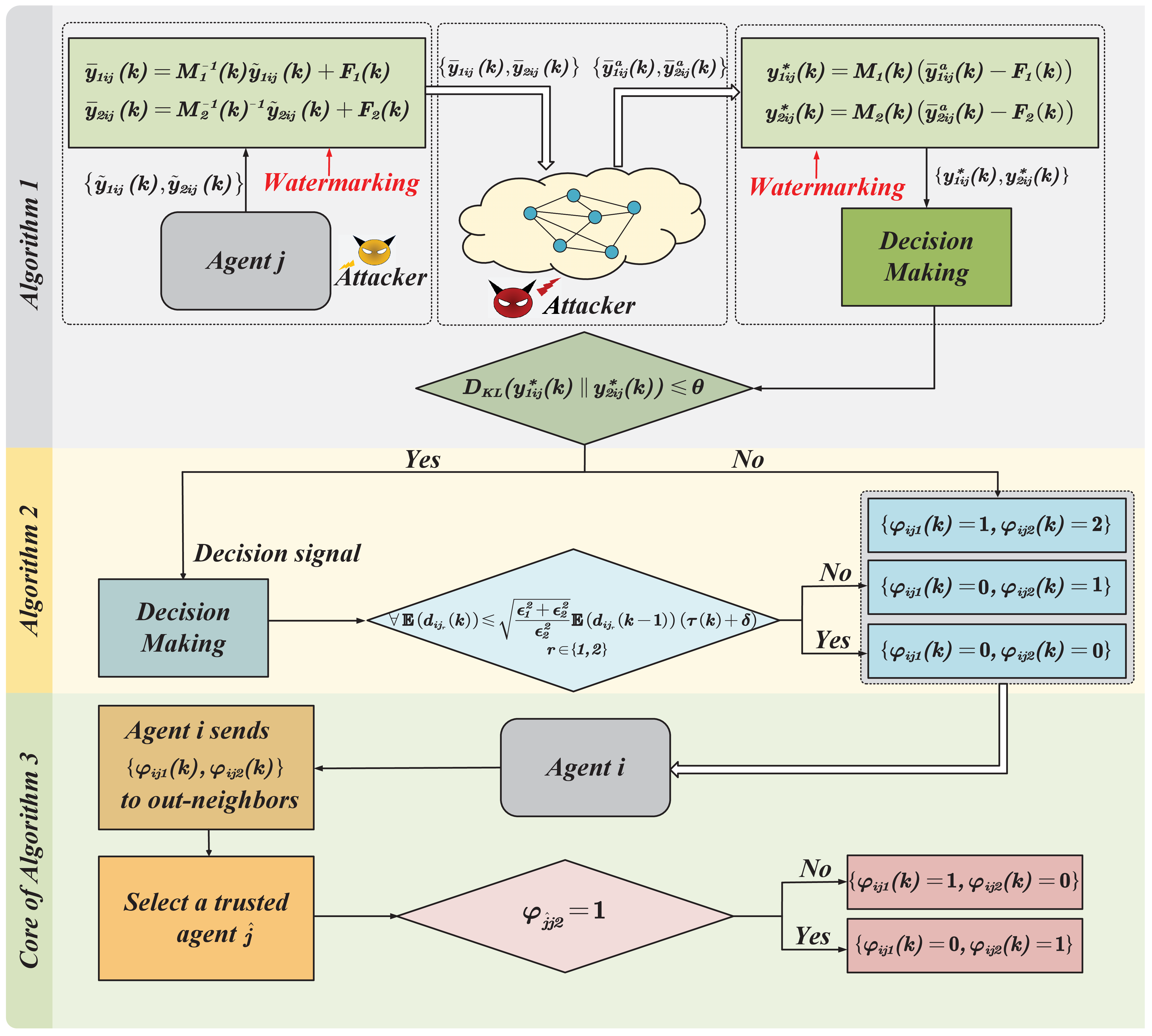}
	\caption{Block diagram of \emph{Algorithm} $\ref{alg3}$ for  hybrid attacks.}
	\label{GOfig66}
\end{figure}

In \emph{Algorithm} $\ref{alg3}$, at instant $k$, for $j \in N_i^{+}$, flag $\left\{\varphi_{ij1}\left( k \right), \varphi_{ij2}\left( k \right) \right\}$ will be generated for the consequence of detection. Specially, $ \varphi_{ij1}\left( k \right)=0, 1, 2$   indicates that the communication channel $(j,i)$ is not attacked, attacked, pending, respectively.   Similarly,  $ \varphi_{ij2}\left( k \right)=0, 1, 2$ stands for  the identity  of agent  $j$ judged by  agent $i$.

To be specific, for each agent $i$ with one of its in-neighbors $j \in N_{i}^{+}$, the attacks on the  communication  channel $(j,i)$ are first determined in    steps 2-8 which corresponds  to \emph{Algorithm} $\ref{alg1}$ to judge whether the message set has been tampered. If the communication channel $(j,i)$ is not falsified,    steps 9-13 will be used in terms of  \emph{Algorithm} $\ref{alg2}$ to judge  whether agent  $j$ is a Byzantine agent. Next, agent  $i$ sends the information  $\left\{ \varphi_{ij1}\left( k \right), \varphi_{ij2}\left( k \right) \right\}$ to its out-neighbors. For the worst   scenario where  the  communication channel $(j,i)$  is attacked, agent $i$ will generate   $\left\{ \varphi_{ij1}\left( k \right)=1, \varphi_{ij2}\left( k \right)=2 \right\}$ for agent  $j$, in which $\varphi_{ij2}\left( k \right)=2$   indicates  that   it is unclear  whether  agent $j$ is attacked or not. Then the flag $\left\{ \varphi_{ij1}\left( k \right)=1, \varphi_{ij2}\left( k \right)=2 \right\}$  is sent to the agents  in  $N_{j}^{-}$.  If there exists an agent $\hat{j}\in N_{j}^{-}$ (trusted agent) satisfying  $\left\{ \varphi_{\hat{j}j1}\left( k \right)=1, \varphi_{\hat{j}j2}\left( k \right)=0 \right\}$ and  $\left\{ \varphi_{i\hat{j}1}\left( k \right)=0, \varphi_{i\hat{j}2}\left( k \right)=0 \right\}$, agent $j$ will be  viewed  as a  Byzantine agent;  otherwise,  a  normal one consistent with  steps 14-23. The following  \emph{Algorithm} $\ref{alg3}$ gives the detailed procedure.

\begin{algorithm}[h]
	\caption{Hybrid Attacks Detection Strategy}
	\label{alg3}
	\begin{algorithmic}[1]
		\State{\textbf{Initialization:} $\left\{ \varphi_{ij1}\left( k \right) =2,\varphi_{ij2}\left( k \right) =2 \right\}, \forall~i,j\in \mathbb{V}$;}
		\State{For each agent  $i$,  agent $j \in {N_{i}^{+}} $  generates   the message set $\left\{\widetilde{y}_{1ij}(k),~\widetilde{y}_{2ij}(k)\right\}$};
		\State{The message set $\left\{\widetilde{y}_{1ij}(k),~\widetilde{y}_{2ij}(k)\right\}$} is equipped with watermarking as  	
		\begin{equation}\label{bbp}
			\begin{aligned}
				&\overline{y}_{1ij}(k)=M_{1}^{-1}(k)\widetilde{y}_{1ij}(k)+F_1(k),\\
				&\overline{y}_{2ij}(k)=M_{2}^{-1}(k)\widetilde{y}_{2ij}(k)+F_2(k);
			\end{aligned}
		\end{equation}	
		\State{The message set $\left\{ \overline{y}_{1ij}\left( k \right) ,\overline{y}_{2ij}\left( k \right) \right\}$ is transmitted into  communication edge $(j,i)$};
		\State{If $\left\{ \overline{y}^{a}_{1ij}\left( k \right) ,\overline{y}^{a}_{2ij}\left( k \right) \right\}$ is received,  the watermarking is removed according to 
			\begin{equation}\label{GOeq22op}
				\begin{aligned}
					&y_{1ij}^{*}(k)=M_1(k)\left( \overline{y}^{a}_{1ij}(k)-F_1\left( k \right) \right),\\
					&y_{2ij}^{*}(k)=M_1(k)\left( \overline{y}^{a}_{2ij}(k)-F_2\left( k \right) \right);
				\end{aligned}
			\end{equation}		
		}       
		\State Decision making:   calculate $D_{KL}(y_{1ij}^{*}(k) \Vert  y_{2ij}^{*}(k))$;
		\If {$D_{KL}(y_{1ij}^{*}(k) \Vert  y^{*}_{2ij}(k))\leqslant\theta$} 
		\State The communication channel $(i,j)$ is  not attacked;
		\If {
			$$\forall~ \mathbb{E}\left( d_{ij_{r}}\left( k \right)\right)  \leqslant\sqrt{\frac{\epsilon_{1}^{2}+\epsilon_{2}^{2}}{\epsilon _{2}^{2}}} \mathbb{E}\left(d_{ij_{r}}\left( k-1 \right)\right)  \left( \tau \left( k \right) +\delta \right),$$ $r\in \left\{ 1,2 \right\} $}
		\State Agent $j$ is  not attacked and   $\left\{ \varphi_{ij1}\left( k \right) =0,\varphi_{ij2}\left( k \right) =0 \right\}$;
		\Else
		\State Agent $j$ is  attacked and   $\left\{ \varphi_{ij1}\left( k \right) =0,\varphi_{ij2}\left( k \right) =1 \right\}$;
		\EndIf
		\Else 
		\State The communication channel $(i,j)$ is   attacked and set $\left\{ \varphi_{ij1}\left( k \right) =1,\varphi_{ij2}\left( k \right) =2 \right\}$;
		\EndIf
		\State Agent $i$ sends $\left\{ \varphi_{ij1}\left( k \right),\varphi_{ij2}\left( k \right) \right\}$ to its out-neighbors;
		\State Select a trusted agent $\hat{j}\in N_i^{+}\cap \left\{ \left. \hat{j}\right| \varphi_{i\hat{j}1}\left( k \right) =0,\varphi_{i\hat{j}2}\left( k \right) =0  \right\} 
		\cap \left\{ \left. \hat{j}\right| \varphi_{\hat{j}j1}\left( k \right) =0 \right\}$;
		\If { $\varphi_{\hat{j}j2}=1$}
		\State Agent $j$ is a  Byzantine agent;
		\Else \State Agent $j$ is a normal agent;
		\EndIf
	\end{algorithmic}
\end{algorithm}

With the help of detection schemes proposed in Subsections $\ref{opp1}$ and $\ref{opp2}$, \emph{Algorithm}  $\ref{alg3}$ provides  a  way  to detect and distinguish the hybrid attacks by introducing the concept of trusted agent $\hat{j}$. The following criterion is given to address this.

\begin{Theorem}\label{GOth4}
	Suppose  that  the conditions in Theorem $\ref{GOth242}$ and Proposition $\ref{GOpro2}$ hold and there exist at  least $L+P+1$ directed two-hop paths  between any pair of neighboring  agents, aligned with the direction of the corresponding edge.
Then, the hybrid attacks with  \emph{Algorithm}  $\ref{alg3}$ are detectable.
\end{Theorem}

\begin{proof}
	We prove this  by  contradiction and  divide the proof procedure into  two cases: for  normal agent $i$, \bm{{Case~A:}} there exists  at least one normal agent $m \in N_{i}^{+}$ or normal edge $(m,i)$	which is    accused of being attacked  mistakenly; \bm{{Case~B:}} there exists  at least one Byzantine agent $\varepsilon  \in N_{i}^{+}$ or edge $(\varepsilon ,i)$ attacked that   is  mistaken as the normal one.
	
	\bm{{Case~A:}} We further separate  it into  two subcases. \bm{{Subcase~A}}$\text{-}$\bm{{1:~}} A normal agent $m$ is mistakenly identified  as the Byzantine one.
	If the normal agent $m$ is judged as the Byzantine agent, according to \emph{Algorithm}  $\ref{alg2}$,  $d_{im}\left( k \right)> d_{im}\left( k-1 \right) \left( \tau \left( k \right) +\delta \right)$. This contradicts   Proposition $ \ref{GOpro2} $ in which  the normal agent satisfies $d_{im}\left( k \right) \leqslant\sqrt{\frac{\epsilon_{1}^{2}+\epsilon_{2}^{2}}{\epsilon _{2}^{2}}}d_{im}\left( k-1 \right) \left( \tau \left( k \right) +\delta \right)$. \bm{{Subcase~A}}$\text{-}$\bm{{2:~}} A normal edge $(i,m)$ is mistakenly believed to be attacked from  communication layer. In this way, it gives  $D_{KL}(y_{1im}^{*}(k) \Vert  y^{*}_{2im}(k))>\theta$. However, for  the attack-free system, based on $ (\ref{GOeq22op})$, we have 
	\begin{equation}\label{GOeq22op}
		\begin{aligned}
			&y_{1im}^{*}(k)=M_1(k)\left( \overline{y}_{1im}(k)-F_1\left( k \right) \right)=\widetilde{y}_{1im}(k),\\
			&y_{2im}^{*}(k)=M_1(k)\left( \overline{y}_{2im}(k)-F_2\left( k \right) \right)=\widetilde{y}_{1im}(k).
		\end{aligned}
	\end{equation}		
	This  indicates that the two state values  in message set transmitted  after the reverse process have  the same distribution.	From  Definition $\ref{GOde1}$, it is obvious that $D_{KL}(y_{1im}^{*}(k) \Vert  y_{2im}^{*}(k))\leqslant \theta$, a  contradiction.
	
	\bm{{Case~B:}} Four subcases will be addressed. First, two single attack scenarios are under consideration.  \bm{{Subcase~B}}$\text{-}$\bm{{1:~}}  Byzantine agent $\varepsilon$ is viewed
	as a  normal one. In this way, we get $d_{i\varepsilon }\left( k \right) \leqslant\sqrt{\frac{\epsilon_{1}^{2}+\epsilon_{2}^{2}}{\epsilon _{2}^{2}}}d_{i\varepsilon}\left( k-1 \right) \left( \tau \left( k \right) +\delta \right)$ which contradicts  Proposition $\ref{GOpro2}$.  \bm{{Subcase~B}}$\text{-}$\bm{{2:~}}  Attacked edge $(\varepsilon,i)$ is regarded as a  normal one. Thus we have $D_{KL}(y_{1ij}^{*}(k) \Vert  y_{2ij}^{*}(k) \leqslant\theta$. However,  in terms of   Theorem $\ref{GOth242}$, the above conclusion becomes invalid with appropriate $\theta$ and  watermarking parameters including  $\sigma _{M_{1}}^{2}$, $\sigma _{F_{2}}^{2}$, $
	\sigma _{M_{2}}^{2}$ and  $\sigma _{F_{1}}^{2}$.
	Next, we consider a  more complex situation that both agent $\omega$ and edge $(\omega,i)$ are suffer from  attacks, which  corresponds  to the hybrid attacks.  
	\bm{{Subcase~B}}$\text{-}$\bm{{3:~}} Attacked edge  $(\omega,i)$ is judged as a  normal one. Similar  to the analysis process, it is derived that the above statement fails. \bm{{Subcase~B}}$\text{-}$\bm{{4:~}} Byzantine agent $\omega$ bypasses the detection strategy in \emph{Algorithm}  $ \ref{alg3} $. If there exist at least  $L+P+1$   two-hop paths  between any two agents, under \emph{Assumption} $\ref{GOas2}$, we can always find   one  trusted  two-hop path  between any two agents  such that edge $(\omega,s)$, agent $s$ and edge $(s,i)$ are safe. Based on the  analysis in \bm{{Subcase~A}}$\text{-}$\bm{{1}}, \bm{{Subcase~A}}$\text{-}$\bm{{2}} and \bm{{Subcase~B}}$\text{-}$\bm{{2}}, we have $\left\{ \varphi_{is 1}\left( k \right) =0,\varphi_{is 2}\left( k \right) =0 \right\}$  and $\left\{ \varphi_{i\omega 1}\left( k \right) =1,\varphi_{i\omega 2}\left( k \right) =2 \right\}$. By resorting to \emph{Algorithm}   $\ref{alg3}$, if agent $\omega$ is regarded as a safe one, it has  $\left\{ \varphi_{s\omega 1}\left( k \right) =0,\varphi_{s\omega 2}\left( k \right) =0 \right\}$.  While, this  violates the conclusion  in  \bm{{Subcase~B}}$\text{-}$\bm{{1}} which should be $\left\{ \varphi_{s\omega 1}\left( k \right) =0,\varphi_{s\omega 2}\left( k \right) =1 \right\}$. The proof is thus  completed.      
	\end {proof}

 {\begin{rmk}
          At present, the detection mechanism based on two-hop communication is mainly targeted at the agent layer. Since this paper studies the hybrid  attacks (including agent  layer and communication layer), we degrade the attacks to agent  layer attacks and compare the detection mechanism. That is,   by fixing  $P=0$, we can compare our results with  the case of  $L$-local \cite{yuan2021secure}. Our algorithm brings some moderate graph 
          requirement and less information transmission. As for the  classical literature  for  two-hop detector \cite{yuan2021secure},   the required  condition about graph   is that:  for   each edge $(i,j)$ in $\mathcal{G}$ and  each agent $h\in N^{+}_{j}$, we have  $h\in N^{+}_{i}$ or there are at least  $2L+1$ directed two-hop edges from  each  agent $h$ to agent $i$.
          In this paper, as shown in Theorem $\ref{GOth4}$, we can see that under the model parameter requirements of  system $(\ref{GOeq1})$, our graph connectivity condition requires only  $L+1$ directed two-hop paths  between any pair of neighboring  agents  with the same direction of the corresponding edge.    
          Moreover, as described in  \cite{yuan2021secure}, it is demonstrated  that within this two-hop communication  detection framework, the transmitted information from each agent $i\in \mathbb{V}$ encompasses its own state value, its own identity  information, the identity and state value of neighboring agents  in the two-hop process, as well as the final result of local detection for malicious agents.	
          However, according to \textit{Algorithm} $\ref{alg3}$, it is  found that the entire network only transmits  the following information: the data set encrypted by the watermarking $\left\{ \overline{y}_{1ij}\left( k \right) ,\overline{y}_{2ij}\left( k \right) \right\}$, and two flag values $\left\{ \varphi_{ij1}\left( k \right),\varphi_{ij2}\left( k \right)\right\}$. Besides, it should be noted  that only  two flag values which belong to  real numbers indicate a two-hop communication. This means  that our findings provide relatively limited amounts of informational content.
	\end{rmk}}

	\section{Platooning  Simulations  }\label{GOsec4}
	
	In this part, we  evaluate the effectiveness of the proposed method  for platooning of connected vehicles  numbered  $0$-$6$ \cite{bian2019reducing}, and   agent $0$ is the  leader. Each vehicle  is regarded as an agent that  communicates with its neighbors via network.  {Fig.
		$\ref{GOfig1224}$ is the information flow topology  for platoons.} 
	Fig.
	$\ref{GOfig2}$ is the communication network topology  under attacks.
	
	 {	\begin{figure*}[htbp]
			\centering
			\includegraphics[width=6in,height=1.4in]{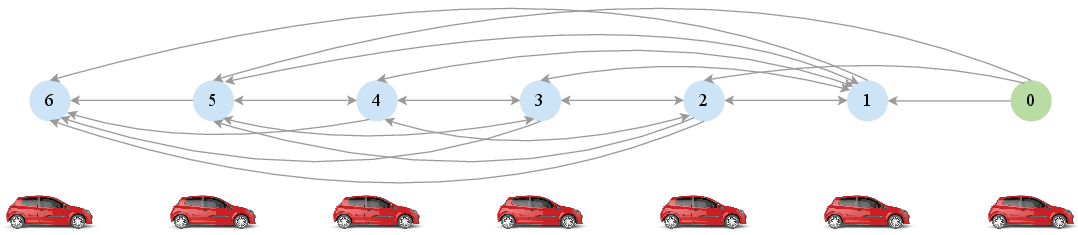}
			\caption{ {Information flow topology for platoons.}}
			\label{GOfig1224}
	\end{figure*} }

	\begin{figure}[H]
		\centering
		\includegraphics[width=2.3in,height=2.5in]{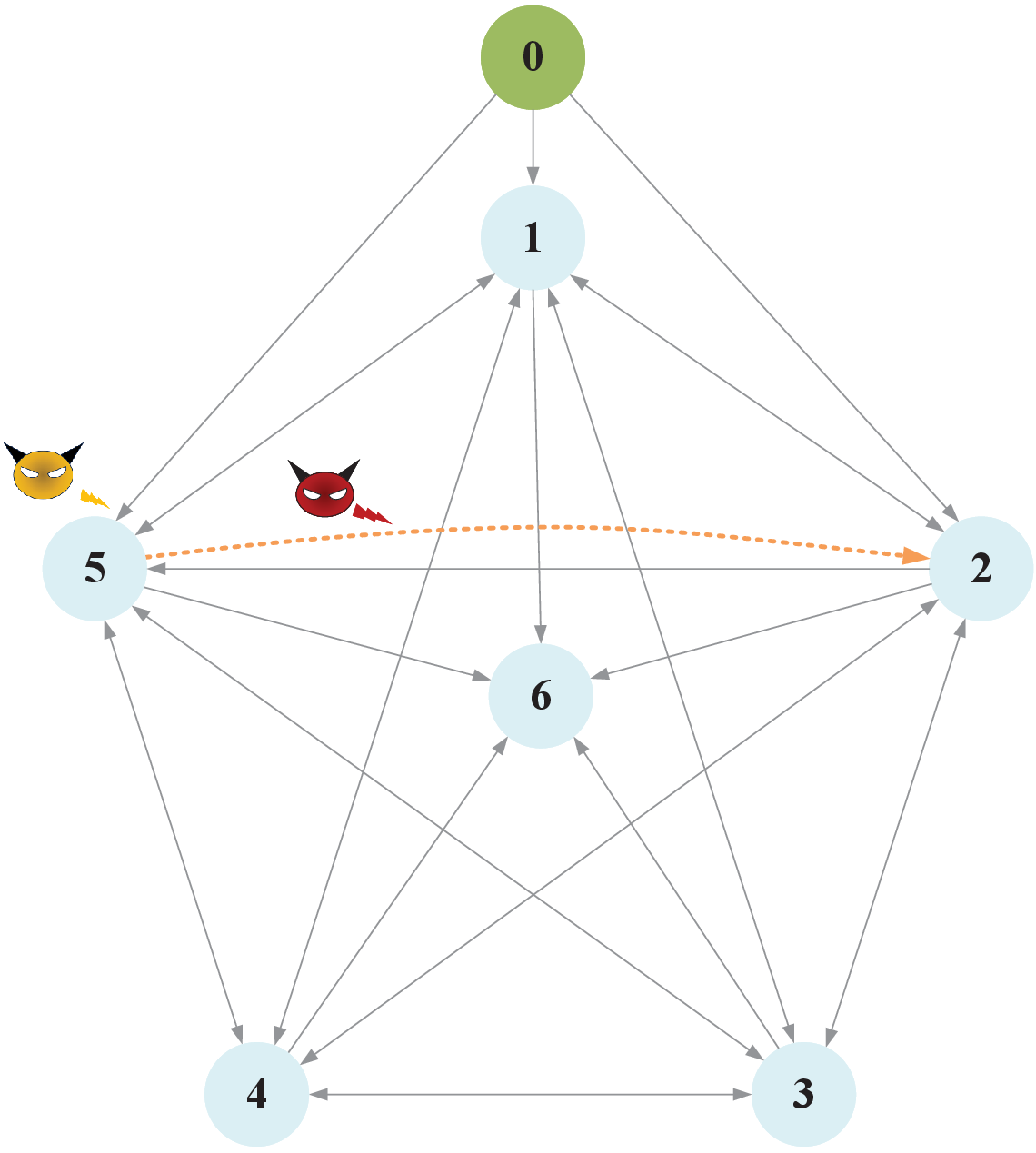}
		\caption{Network topology for platooning of connected vehicles.}
		\label{GOfig2}
	\end{figure}

	The dynamics of  vehicle model $i$ have the form  
	$$
	\begin{cases}
		\dot{p}_i\left( t \right) =v_i\left( t \right),\\
		\dot{v}_i\left( t \right) =a_i\left( t \right),\\
		\varDelta \dot{a}_i\left( t \right) +a_i\left( t \right) =u_i\left( t \right),\\
	\end{cases}
	$$
	where $p_i\left( k \right) \in \mathbb{R}$, $v_i\left( k \right) \in \mathbb{R}$, $a_i\left( k \right) \in \mathbb{R}$ and $\varDelta \in \mathbb{R}$ are   position, velocity,  acceleration and   inertial time lag in the powertrain, respectively. 
	
	The  discretized version of the above model is  
	
	$$
	\left[ \begin{array}{c}
		p_i\left( k+1 \right)\\
		v_i\left( k+1 \right)\\
		a_i\left( k+1 \right)\\
	\end{array} \right] =\left( \bm{I}+A \right) \left[ \begin{array}{c}
		p_i\left( k \right)\\
		v_i\left( k \right)\\
		a_i\left( k \right)\\
	\end{array} \right] +Bu_i\left( k \right), 
	$$\\
	where 
	$$
	A=\left[ \begin{matrix}
		0&		T&		0\\
		0&		0&		T\\
		0&		0&		-\frac{T}{\varDelta}\\
	\end{matrix} \right], 
	\\
	B=\left[ \begin{matrix}
		0\\		0\\		1\\
	\end{matrix} \right]$$
	and $T  \in \mathbb{R}^{+}$ is the sampling period.
	Here  $\varDelta=1.2$, $T=1$.
	
	For attacks on the communication layer, $100$ Monte Carlo trials are carried out. The parameters of   additive and multiplicative watermarking are set  as $\lambda_1=2$, $\lambda_2=5$ $\sigma _{M_{1}}^{2}=7.2$, $\sigma _{M_{2}}^{2}=4.3$, $\sigma _{F_{1}}^{2}=2$, $\sigma _{F_{2}}^{2}=3.5$ and $\sigma _{1}^{2}=4$ as well as  the detection threshold   $\theta =4.61$.
	
	Fig. $\ref{GOfig13}$ depicts  the transient performance comparison for different  detection schemes at $k=4$,  which corresponds  to the  transient  stage  satisfying  $\eta(k) \geqslant 0.05
	$. We can see that for  the normal system,  the maximum KL    divergence monotonically increases with  ${\max}\left\| d_{ij}\left( 0 \right) \right\|$ in \cite{mustafa2020resilient}. However, the KL  divergence under  \emph{Algorithm}  $\ref{alg1}$ almost  keeps unchanged  which indicates that our detection scheme is more suitable for transient stage.       When the system is attacked on   the communication  layer from $k\geqslant 10$  at which  the system has not reached  steady state, the attack signals   $\overline{y}_{125}\left(k \right)=\varXi _{125}(k) \overline{y}_{125}(k)+\varLambda_{125}(k)$ and $\overline{y}_{225}\left(k \right)=\varXi _{225}(k) \overline{y}_{225}(k)+\varLambda_{225}(k)$
	are  injected  into   edge $(5,2)$ where the forms of $\varXi _{125}(k), \varLambda_{125}(k), \varXi _{225}(k)$ and $\varLambda_{125}(k)$ are listed below. Fig. $\ref{GOfig455}$ plots  the  KL divergence of the system under such an  attack. It can be seen that the alarm of  detector associated with  edge $(5,2)$  is triggered  after  $k=10$.

	$$
	\begin{cases}
		\varXi  _{125}\left( k \right) =\left[ \begin{matrix}
			\sin\mathrm{(}k)&		0&		0\\
			0&		8.3\sin\mathrm{(}k)&		0\\
			0&		0&		2.4\sin\mathrm{(}k)\\
		\end{matrix} \right],\\
		\varLambda _{125}\left( k \right) =\left[ \begin{array}{c}
			0\\
			3.73\sin\mathrm{(}k)\\
			-1.32\sin\mathrm{(}k)\\
		\end{array} \right] .\\
	\end{cases}
	$$
	$$
	\begin{cases}
		\varXi  _{225}\left( k \right) =\left[ \begin{matrix}
			0&		0&		0\\
			0&		7.3\sin\mathrm{(}k)&		0\\
			0&		0&		-2.32\sin\mathrm{(}k)\\
		\end{matrix} \right],\\
		\varLambda _{225}\left( k \right) =0.\\
	\end{cases}
	$$

	\begin{figure}[H]
		\centering
		\includegraphics[width=2.5in,height=2.0in]{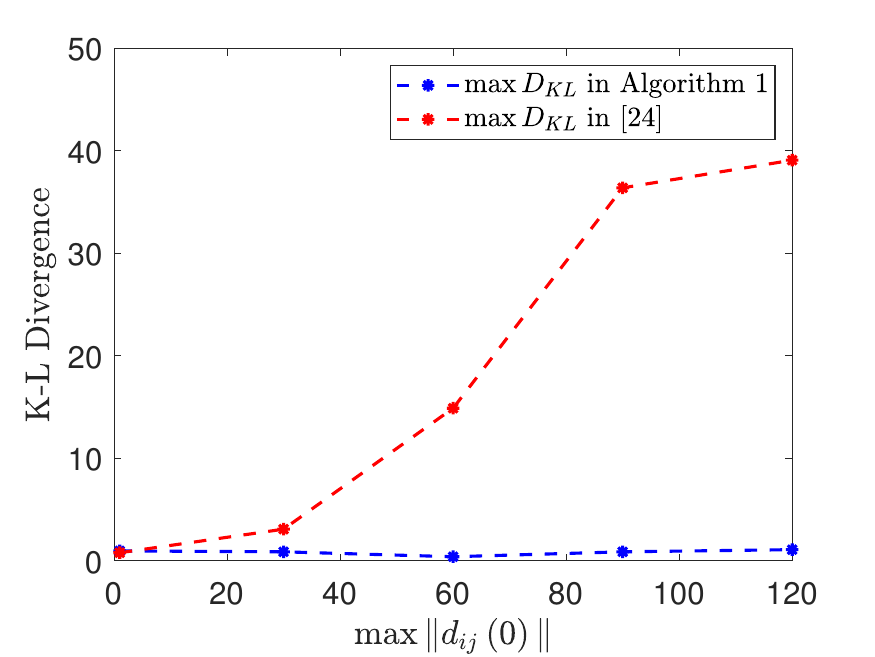}
		\caption{The maximum KL divergence curves after  $k=4$ under different initial relative errors ${\max}\left\| d_{ij}\left( 0 \right) \right\|$ of detection strategy in \cite{mustafa2020resilient} and {Algorithm}  $\ref{alg1}$ in this paper.}
		\label{GOfig13}
	\end{figure}

	\begin{figure}[H]
		\centering
		\includegraphics[width=2.5in,height=2.0in]{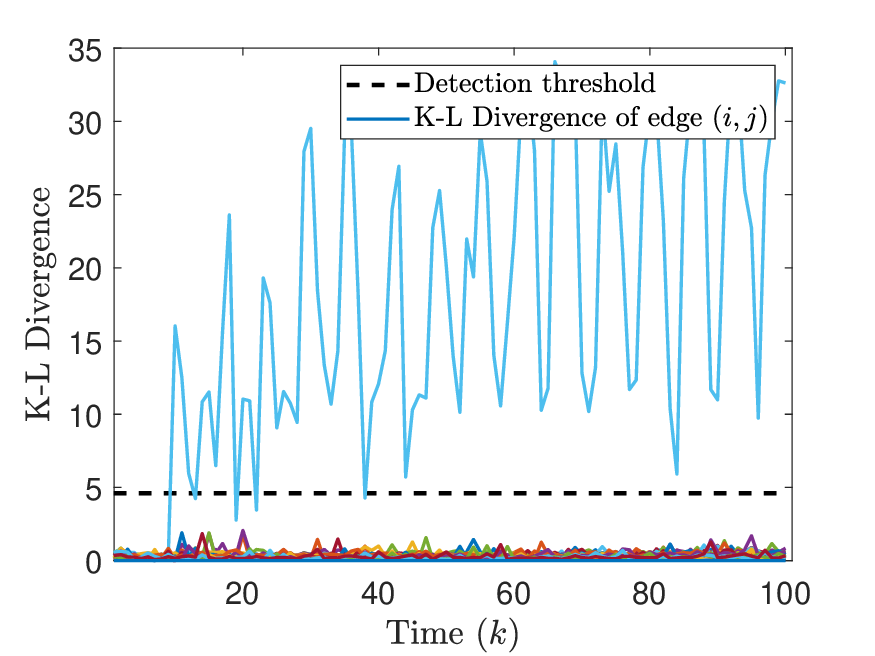}
		\caption{KL divergence of  system suffering from attacks on the  communication  layer under {Algorithm} 1.}
		\label{GOfig455}
	\end{figure}

	For    attacks  acting on  the agent layer, we set $M_r=100$, $\phi=0.16$ and $\delta=6$. From Fig. $\ref{GOfig5}$, when there is no attack, we can see that $\left\| y_{ij}\left( k \right)-x_{i}\left( k \right)\right\|$  is below the envelope, and the detector alarm will not be triggered. When the system is attacked for  $k\geqslant 20$ in which the Byzantine  agent  is  $5$, Fig. $\ref{GOfig6}$ gives the corresponding curves of  $ \left\| y_{ij}\left( k \right)-x_{i}\left( k \right)\right\|$.  It is found  that  the detector alarm is triggered immediately.
	
	\begin{figure}[H]
		\centering
		\includegraphics[width=2.5in,height=2.0in]{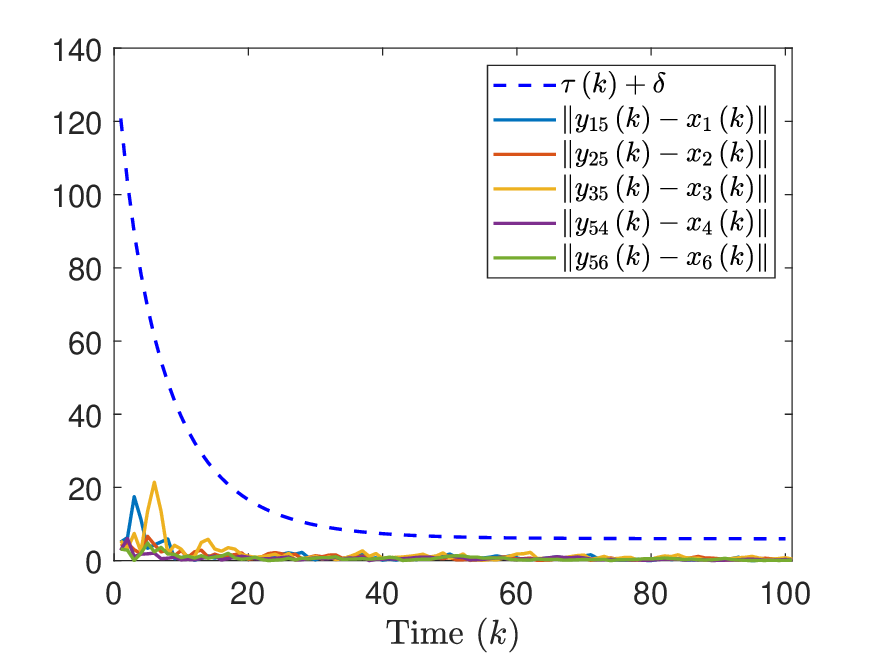}
		\caption{$\left\| y_{ij}\left( k \right)-x_{i}\left( k \right)\right\|$ of normal system under {Algorithm}  $\ref{alg2}$.}
		\label{GOfig5}
	\end{figure}
	
	\begin{figure}[H]
		\centering
		\includegraphics[width=2.5in,height=2.0in]{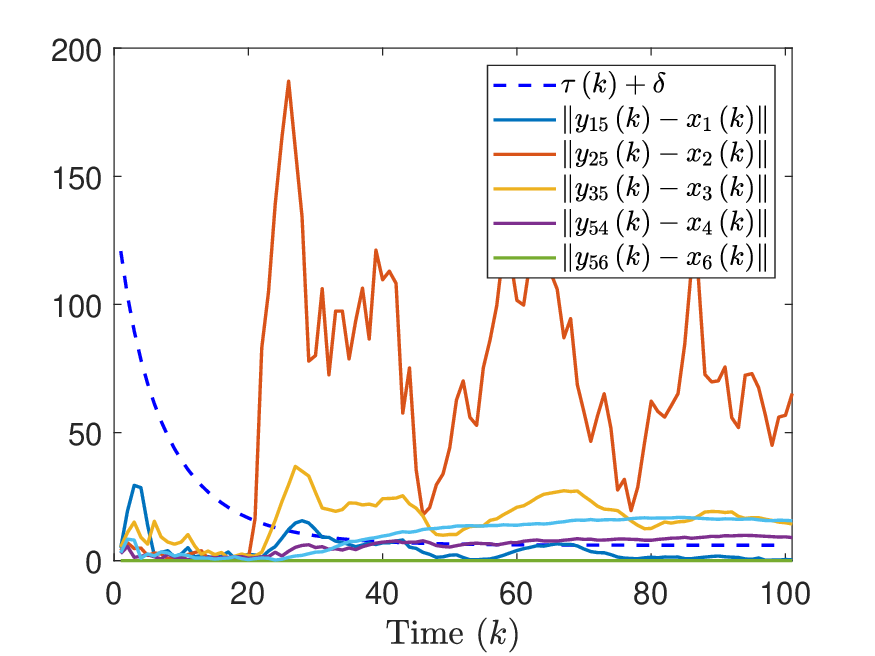}
		\caption{$ \left\| y_{ij}\left( k \right)-x_{i}\left( k \right)\right\|$  suffering  Byzantine attacks  under {Algorithm}  $\ref{alg2}$.}
		\label{GOfig6}
	\end{figure}

	For the  hybrid attacks, we consider the scenario  that attacks are implemented on  agent  $5$ and edge $(5,2)$. In order to make the simulations  clearer, three attack conditions  are addressed, e.g., for $k \in [2,4)$, attacks on the  communication  layer are   implemented on edge $(5,2)$; for $k \in [4,6)$, hybrid attacks are injected into agent  $5$ and edge $(5,2)$;  for  $k \in [6,8)$, only attacks on agent $5$ are  imposed. TABLE $\rm{\ref{GOfi99}}$ lists the   detailed results. It can be observed that,  over all times, agent $2$ judges the flag value of the two hop agent $\hat{j}\in \left\{1,3,4\right\} $ between agent $2$ and agent $5$ as $\left\{ \varphi_{2\hat{j}1}\left( k \right)=0, \varphi_{2\hat{j}2}\left( k \right)=0 \right\}$ which means  that agents  $1$, $3$ and $4$ are trusted agents. More specific,

	\begin{enumerate}
		\item {Attacks on the  communication layer:} For $k \in [2,4)$, the flag  judged by agent $2$ to agent $5$ is $\left\{ \varphi_{251}\left( k \right)=1, \varphi_{252}\left( k \right)=2 \right\}$, indicating   edge $(5,2)$ is attacked. And  the  flag incured  by agent  $\hat{j}$ to agent $5$ becomes  $\left\{ \varphi_{\hat{j}51}\left( k \right)=0, \varphi_{\hat{j}52}\left( k \right)=0 \right\}$,  which implies that  agent $5$ is not affected  by the  attacks on the  agent layer. That is to say,  agent $5$ suffers from attacks  only  on the communication  layer.  
		\item Hybrid attacks: For $k \in [4,6)$, the flag caused  by agent $2$ to agent $5$ turns to be  $\left\{ \varphi_{251}\left( k \right)=1, \varphi_{252}\left( k \right)=2 \right\}$, which means that    edge $(2,5)$ is attacked. Moreover, the flag  judged by agent  $\hat{j}$ to agent $5$ is $\left\{ \varphi_{\hat{j}51}\left( k \right)=0, \varphi_{\hat{j}52}\left( k \right)=1 \right\}$,  indicating that agent $5$ is  attacked. Thus  agent $5$ is under  hybrid attacks. 
		\item Attacks on the agent layer: For $k \in [6,8)$, the flag induced by agent $2$ to agent $5$ has the form  $\left\{ \varphi_{251}\left( k \right)=0, \varphi_{252}\left( k \right)=0 \right\}$, showing  edge $(5,2)$ has not  attacked. In addition, the flag  judged by agent  $\hat{j}$ to agent $5$ is $\left\{ \varphi_{\hat{j}51}\left( k \right)=0, \varphi_{\hat{j}52}\left( k \right)=1 \right\}$  which implies  that agent $5$ is  attacked. We can say that  there are attacks for agent $5$ only on the agent layer.
	\end{enumerate}

	\begin{table}[ht]
		\caption{The flags of {Algorithm}  $\ref{alg3}$ against hybrid attacks}
		\centering
		
		\renewcommand\arraystretch{1.6}{
			\setlength{\tabcolsep}{1mm}{\scalebox{0.9}{
					\begin{tabular}{|c|c|c|c|}
						\hline
						{\diagbox[innerwidth=3.3cm]{Flags}{Time$(k)$}}&$[2,4)$&$[4,6)$&$[6,8)$\\ \hline
						$\left\{\varphi_{201}\left( k \right),\varphi_{202}\left( k \right)\right\}$& $\left\{0,0\right\}$&$\left\{0,0\right\}$&$\left\{0,0\right\}$ \\ \hline 
						$\left\{\varphi_{211}\left( k \right),\varphi_{212}\left( k \right)\right\}$& $\left\{0,0\right\}$&$\left\{0,0\right\}$&$\left\{0,0\right\}$\\ \hline
						$\left\{\varphi_{231}\left( k \right),\varphi_{232}\left( k \right)\right\}$& $\left\{0,0\right\}$&$\left\{0,0\right\}$&$\left\{0,0\right\}$\\ \hline
						$\left\{\varphi_{241}\left( k \right),\varphi_{242}\left( k \right)\right\}$& $\left\{0,0\right\}$&$\left\{0,0\right\}$&$\left\{0,0\right\}$\\ \hline 
						$\left\{\varphi_{251}\left( k \right),\varphi_{252}\left( k \right)\right\}$& $\left\{1,2\right\}$&$\left\{1,2\right\}$&$\left\{0,1\right\}$\\ \hline
						$\left\{\varphi_{151}\left( k \right),\varphi_{152}\left( k \right)\right\}$& $\left\{0,0\right\}$&$\left\{0,1\right\}$&$\left\{0,1\right\}$\\ \hline    
						$\left\{\varphi_{351}\left( k \right),\varphi_{352}\left( k \right)\right\}$& $\left\{0,0\right\}$&$\left\{0,1\right\}$&$\left\{0,1\right\}$\\ \hline
						$\left\{\varphi_{451}\left( k \right),\varphi_{452}\left( k \right)           \right\}$& $\left\{0,0\right\}$&$\left\{0,1\right\}$&$\left\{0,1\right\}$\\ \hline
						$\left\{\varphi_{211}\left( k \right),\varphi_{212}\left( k \right)\right\}$& $\left\{0,0\right\}$&$\left\{0,0\right\}$&$\left\{0,0\right\}$\\ \hline    
						$\left\{\varphi_{231}\left( k \right),\varphi_{232}\left( k \right)\right\}$& $\left\{0,0\right\}$&$\left\{0,0\right\}$&$\left\{0,0\right\}$\\ \hline
						$\left\{\varphi_{241}\left( k \right),\varphi_{242}\left( k \right)           \right\}$& $\left\{0,0\right\}$&$\left\{0,0\right\}$&$\left\{0,0\right\}$\\ \hline
		\end{tabular}}}}
		\label{GOfi99}
	\end{table}

	\section{Conclusion}\label{GOsec5}
	This paper has proposed a  detection framework for single/hybrid attacks on the communication layer and agent layer. It can handle both transient and steady stages. For  attacks on the  communication  layer,  a KL divergence detection algorithm has been given  based on a modified  watermarking, as well as  a sufficient condition  ensuring  the detection capability. For  attacks on the agent layer,  a  detection scheme involving  the convergence rate has been  designed, which can enhance  the resilience of the system against  attacks while  guaranteeing the detection effect. For hybrid attacks (i.e. attacks on both the communication  layer and  agent  layer),   a general detection framework in terms of  trusted agents has been  suggested, which can detect and  locate attacks accurately, while relaxing the requirements of the  graph. In the future, we will focus on the
	resilient  control of systems under different kinds of attacks.

\section*{References}

\def\refname{\vadjust{\vspace*{-2.5em}}}

\end{document}